\documentclass[11pt,a4paper]{article}

\usepackage{mathtools}
\usepackage{authblk} 
\usepackage{algpseudocode,algorithmicx,algorithm}

\usepackage{mathrsfs}
\usepackage{latexsym,bm}
\usepackage{amsmath,amsfonts,amsmath,amssymb,amsthm}
\usepackage{extarrows}

\usepackage{graphicx,subfigure,epstopdf,float}
\usepackage{enumerate,cases,multirow}
\usepackage{makecell}
\usepackage{caption}
\usepackage{tikz}
\usepackage{pgfplots}

\usepackage{longtable,colortbl,arydshln,threeparttable}
\definecolor{mygray}{gray}{.9}
\usepackage[shortlabels]{enumitem}

\usepackage{indentfirst}
\usepackage[top=25mm,bottom=20mm,left=25mm,right=20mm]{geometry}
\usepackage{cite}

\usepackage{listings}

\usepackage{makeidx}        
\usepackage{booktabs}
\usepackage[bookmarksnumbered,colorlinks,citecolor=red,linkcolor=red,hyperindex,linktocpage=true]{hyperref}

\newcommand{\Ba}{{\boldsymbol{a}}}
\newcommand{\Bb}{{\boldsymbol{b}}}
\newcommand{\Bc}{{\boldsymbol{c}}}

\newcommand{\Bj}{{\boldsymbol{j}}}

\newcommand{\Bn}{{\boldsymbol{n}}}

\newcommand{\Bu}{{\boldsymbol{u}}}
\newcommand{\Bv}{{\boldsymbol{v}}}
\newcommand{\Bw}{{\boldsymbol{w}}}
\newcommand{\Bx}{{\boldsymbol{x}}}

\newcommand{\BB}{{\boldsymbol{B}}}
\newcommand{\BC}{{\boldsymbol{C}}}
\newcommand{\BD}{{\boldsymbol{D}}}
\newcommand{\BE}{{\boldsymbol{E}}}
\newcommand{\BF}{{\boldsymbol{F}}}

\newcommand{\BH}{{\boldsymbol{H}}}

\newcommand{\BJ}{{\boldsymbol{J}}}

\newcommand{\BM}{{\boldsymbol{M}}}
\newcommand{\BN}{{\boldsymbol{N}}}

\newcommand{\BP}{{\boldsymbol{P}}}

\newcommand{\FJ}{\mathfrak{J}}

\newcommand{\Cb}{\mathcal{B}}
\newcommand{\Cc}{\mathcal{C}}

\newcommand{\Ce}{\mathcal{E}}
\newcommand{\Cf}{\mathcal{F}}

\newcommand{\Cj}{\mathcal{J}}

\newcommand{\Cm}{\mathcal{M}}

\newcommand{\bbB}{\mathbb{B}}
\newcommand{\bbC}{\mathbb{C}}

\newcommand{\bbR}{\mathbb{R}}

\newcommand{\bbP}{\mathbb{P}}
\newcommand{\bbQ}{\mathbb{Q}}

\newcommand{\bbV}{\mathbb{V}}
\newcommand{\D}{\mathrm{d}}

\newcommand{\ket}[1]{| #1 \rangle} 
\newcommand{\bra}[1]{\langle #1 |} 

\newcommand{\vect}{\boldsymbol}

\newcounter{parentalgorithm}

\makeatother

\newtheorem{theorem}{Theorem}[section]
\newtheorem{lemma}{Lemma}[section]

\theoremstyle{remark}

\theoremstyle{assumption}

\numberwithin{equation}{section}

\begin{document}

	\title{Quantum simulation of Maxwell's equations via Schr$\ddot{\text{o}}$dingerisation}

\author[1,2,3,5]{Shi Jin   \thanks{shijin-m@sjtu.edu.cn}}
\author[2,3,4,5]{Nana Liu   \thanks{nana.liu@quantumlah.org}}
\author[1]{Chuwen Ma   \footnote{Corresponding author.} \thanks{chuwenii@sjtu.edu.cn}}
	\affil[1]{School of Mathematical Sciences,   Shanghai Jiao Tong University, Shanghai 200240, China.}
	\affil[2]{Institute of Natural Sciences, Shanghai Jiao Tong University, Shanghai 200240, China.}
    \affil[3]{Ministry of Education, Key Laboratory in Scientific and Engineering Computing, Shanghai Jiao Tong University, Shanghai 200240, China.}
  \affil[4]{University of Michigan-Shanghai Jiao Tong University Joint Institute, Shanghai 200240, China.}  
\affil[5]{Shanghai Artificial Intelligence Laboratory, Shanghai, China}

\date{}

	\maketitle
	\begin{abstract}
We present  quantum algorithms for electromagnetic fields governed by Maxwell's equations. The algorithms are based on the Schr\"odingerisation approach, which transforms any linear PDEs and ODEs with non-unitary dynamics into a system evolving under unitary dynamics, via a warped phase transformation that maps the equation into one higher dimension. In this paper, our quantum algorithms are based on either a direct approximation of Maxwell's equations combined with Yee's algorithm, or a matrix representation in terms of Riemann-Silberstein vectors combined with a spectral approach and an upwind scheme. We implement these algorithms with physical boundary conditions, including perfect conductor and impedance boundaries. We also solve Maxwell's equations for a linear inhomogeneous medium, specifically the interface problem. Several numerical experiments are performed to demonstrate the validity of this approach. In addition, instead of qubits, the quantum algorithms can also be formulated in the continuous variable quantum framework, which allows the quantum simulation of Maxwell's equations in analog quantum simulation.  
\end{abstract}


\section{Introduction}
 It has been pointed out in \cite{Fre69} that classical computing technology, which has been developing rapidly for more than half a century, could soon reach its limit as prescribed by the laws of physics. In 1982, Feynman \cite{Feyman82} proposed a new type of computer called a quantum computer that could simulate the physical world more efficiently than conventional computers.
     Deutsch provided a quantum-mechanical model of the theory of quantum computing \cite{Dutsch85}, and it has been shown that quantum computers could potentially outperform the most powerful classical computers  for certain types of problems \cite{Ekert98,Nielssen2000}.
     Today, quantum computing is considered a promising candidate to overcome the limitations of classical computing\cite{Divi95,Shor94,Stea98}.
     
    The application of quantum algorithms to Maxwell's equations has already been discussed in the literature, which can be classified into two main approaches. One is to solve the linear system of equations using the Harrow-Hassidim-Lloyd (HHL) algorithm, which 
    is combined with finite element method (FEM) \cite{cla13,Zhang21}, or Methods of Moments \cite{cai21}.
    The other is to rewrite the source-free Maxwell formulation into a Hamiltonian system based on the Riemann-Silberstein vectors for simulation \cite{costa19,suau21,vaha2020,vaha2020b,va20,va202,Bui2022}. 
    Since the Hamiltonian results from a tensor product of Pauli matrices, 
    it is easy to implement the quantum lattice algorithm (QLA) \cite{vaha2020,vaha2020b,va20,va202}, following 
     \cite{Yepe02}.
    Despite the intensive research on quantum computing in recent years, study on the quantum simulation of electromagnetic models with physical boundary conditions and complex medium are scarce in the literature.

    In the present paper, we investigate the use of another method -- the Schr\"odingeri-sation 
    method -- which opens many new opportunities for quantum simulation of complex physical systems, in both qubit \cite{JLY22a,JLY22b} and continuous-variable frameworks \cite{CVPDE2023}. The main idea is to convert 
    linear partial differential equations (PDEs) or  ordinary differential equations (ODEs) with non-unitary dynamics to a system of Schr$\ddot{\text{o}}$dinger type equations -- with unitary dynamics -- by a warped phase transformation. This transformation converts the non-unitary dynamics into a unitary one by introducing one auxiliary space-like dimension. 
     The method can also be extended to solve open quantum systems in a bounded domain with artificial boundary conditions \cite{JLLY22}, and physical and interface conditions \cite{JLLY23}.

     Our algorithm is an extension of the Schr$\ddot{\text{o}}$dingerisation method applied to Maxwell's equations with suitable boundary conditions. The main contributions are the following.
     \begin{enumerate}
     	\item [(a)] We apply the Schr$\ddot{\text{o}}$dingerisation approach combined with the spectral method or the upwind scheme to the eight-dimensional matrix representation of Maxwell's equations.  We then construct a quantum algorithm for Yee's scheme \cite{TafYee00,Cai13} to simulate the original electromagnetic equations. Yee's algorithm is the most popular algorithm for numerically approximating Maxwell's equations, due to its simplicity and preservation of the continuous vector identities on the discrete grid.
     	We compare these two methods theoretically and numerically. We then analyze the influence of Schr$\ddot{\text{o}}$dingerisation  on the traditional numerical methods for preserving the divergence-free condition of the magnetic field and the total energy of the electromagnetic field. 
      
     	\item [(b)] We convert the boundary and interface conditions of the original electromagnetic model into the conditions of the matrix representation of Maxwell's equations via a unitary transformation. We then give the implementation details of the application of the Schr$\ddot{\text{o}}$dingerisation method to complex boundary value problems and interface problems for Maxwell's equations.
       
     	\item [(c)] We apply the  Schr$\ddot{\text{o}}$dingerisation method 
     	to continuous-variable quantum systems introduced in \cite{CVPDE2023} to 
     	solve the  Maxwell equations. This approach avoids the dense Hamiltonian matrix caused by the discretisation of the velocity varying with space, and allows one to use analog quantum computing to solve the system, which may be more accessible for intermediate-term devices. 
     \end{enumerate} 
     
     The rest of the paper is organized as follows. In Section~\ref{sec:A brief review of Maxwell equations}, we give a brief review of Maxwell‘s equations.
     In Section~\ref{sec:schr}, we review the Schr$\ddot{\text{o}}$dingerisation
     approach for general linear systems. In Section~\ref{sec:schr periodic BD} and \ref{sec:schr physical boundary conditions}, we give implementation details for the three boundary conditions, including periodic, perfect conductor and impedance boundary conditions \cite{ACL18}.
     In Section~\ref{sec:schr inhomogeneous medium}, we apply the Schr$\ddot{\text{o}}$dingerisation method to interface problems.
     In Section~\ref{sec:continuous-variable formulation}, we show the Schr\"odingerisation framework in the continuous-variable representation to simulate Maxwell's equations.
     Finally, we show the numerical tests in Section ~\ref{sec:numerical simulation}.

	Throughout the paper, we restrict the simulation to a finite time interval $t\in [0,T]$, and we use a 0-based indexing, i.e. $j= \{0,1,\cdots,N-1\}$, or $j\in [N]$, and  $\ket{j}\in \bbC^N$, denotes a vector with the $j$-th component being $1$ and others $0$.
	We shall denote the identity matrix and null matrix by $\textbf{1}$ and $\textbf{0}$, respectively,
	and the dimensions of these matrices should be clear from the context, otherwise,
	the notation $\textbf{1}_N$ stands for the $N$-dimensional identity matrix.

	\section{A brief review of Maxwell's equations} \label{sec:A brief review of Maxwell equations}
  
	We consider the Maxwell equations for a medium, in presence of sources of charge $\rho$ and currents $\BJ$,
	\begin{equation}\label{eq:maxwell linear medium}
		\begin{aligned}
			\frac{\partial \BD}{\partial t}- \nabla \times \BH &=-\BJ,\qquad
			\frac{\partial \BB}{\partial t}+ \nabla \times \BE  =0,\\
			\nabla \cdot \BB &=0, \qquad \quad 
			\nabla \cdot \BD  = \rho,
		\end{aligned}
	\end{equation} 
	in the three-dimensional domain $\Omega = [0,1]^3$.
	Assuming the medium to be linear, 
	the electric field, $\BE(x,t)$, the electric flux, $\BD(x,t)$, as well as the magnetic field, $\BH(x,t)$, and the magnetic flux density, $\BB(x,t)$, are related 
	through the constitutive relations
	\begin{equation}\label{eq:maxwell relation}
		\BD = \varepsilon \BE,\quad \BH = \BB/\mu, 
	\end{equation} 
	where $\varepsilon$ and $\mu$ are the permittivity and permeability of the medium, respectively, and they may vary with space and time.
	
	Maxwell’s equations must be supplemented by boundary conditions that must be
	satisfied by the electric and magnetic fields at physical boundaries.
	For the important special case of a perfect conductor, the conditions take 
	a special form as the perfect conductor supports surface charges and 
	currents, whereas the fields are unable to penetrate into the body \cite{ACL18}, i.e.,
	\begin{equation}\label{eq:EB perfect conductor}
		\Bn \times \BE = \textbf{0},\qquad  \Bn\cdot \BB=0, \quad \text{on}\; \partial \Omega,
	\end{equation}
	where $\Bn$ is the unit normal to the boundary $\partial \Omega$. 
	However, there also exist media that are more or less dissipative, for
	instance, when the exterior medium $\bbR^3\backslash \Omega$ is a conductor but not a perfect one \cite{ACL18}. In this case, an impedance boundary condition appears in which the tangential electric and magnetic fields are related through a surface impedance $Z_s$,
	\begin{equation}\label{eq:EB impedance boundary condition}
		\BE\times \Bn +Z_s \Bn \times (\BH \times \Bn) =0, \quad \text{on}\; \partial \Omega.
	\end{equation} 
	In its simplest form, the impedance $Z_s=\sqrt{\mu/\varepsilon}$ is a characteristic of the medium, which allows the plane wave to leave the domain $\Omega$ with velocity $v=1/\sqrt{\mu \varepsilon}$ if $\partial \Omega$ is a plane.
	
	\subsection{A matrix representation of Maxwell's equations }
 
	We shall now consider a medium in which $\varepsilon$ and $\mu$ are independent of time. The electromagnetic equations \eqref{eq:maxwell linear medium}-\eqref{eq:maxwell relation} 
	can be written with the unknowns $\BE$ and $\BB$. They read as
	\begin{equation} \label{eq:maxwell}
		\begin{aligned}
			\frac{\partial}{\partial t}(\sqrt{\varepsilon} \BE)
			-v(\nabla-\frac{1}{2\mu} \nabla \mu)\times (\frac{\BB}{\sqrt{\mu}})
			&= - \frac{\BJ}{\sqrt{\varepsilon}},\quad
			(\nabla +\frac{1}{2 \mu} \nabla \mu)\cdot (\frac{\BB}{\sqrt{\mu}})=0,\\
			\frac{\partial}{\partial t}(\frac{\BB}{\sqrt{\mu}})
			+v(\nabla -\frac{1}{2\varepsilon} \nabla \varepsilon)\times (\sqrt{\varepsilon}\BE) &=0,\quad
			(\nabla +\frac{1}{2 \varepsilon} \nabla \varepsilon)\cdot (\sqrt{\varepsilon}\BE)=\frac{\rho}{\sqrt{\varepsilon}},
		\end{aligned}
	\end{equation}
	where $v = 1/\sqrt{\varepsilon \mu }$ is the speed of light in the medium.
	Remark that vacuum is a particular case of a homogeneous medium. 
	For the sake of simplicity of notation, define
 \begin{align*}
     \Cf &= \frac{1}{\sqrt{2}} \bigg(
     \sqrt{\varepsilon} E_x \ket{0}
     +\sqrt{\varepsilon} E_y \ket{1}
     +\sqrt{\varepsilon} E_z \ket{2}
     +\frac{1}{\sqrt{\mu}} B_x \ket{4}
     +\frac{1}{\sqrt{\mu}} B_y \ket{5}
     +\frac{1}{\sqrt{\mu}} B_z \ket{6}
     \bigg),\\
     \Cj &= \frac{1}{\sqrt{2\varepsilon}} \bigg(
     J_x \ket{0}+J_y \ket{1}+J_z\ket{2}
     -v\rho\ket{7}
     \bigg).
 \end{align*}
	Write Equation~\eqref{eq:maxwell} in vector form as
	\begin{align}
		&\frac{\partial \mathcal{F}}{\partial t} = \Cm \Cf -\Cj= \begin{bmatrix}
			\textbf{0} &v\Cm_{12}\\
			v\Cm_{21} &\textbf{0}
		\end{bmatrix}  \Cf-\Cj. \quad 
    \label{eq:maxwell matrix 1}  
    \end{align}
     Here the operator $\Cm_{12}$ and $\Cm_{21}$ are defined by
     \begin{align*}
  & \Cm_{12} = \Cc-\Cc\bar{\mu}, \quad
		\Cm_{21} = -\Cc+\Cc\bar{\varepsilon},\quad
  \Cc = \begin{bmatrix}
			0 &-\partial_z &\partial_y &-\partial_x\\
			\partial_z &0 &-\partial_x &-\partial_y\\
			-\partial_y &\partial_x &0 &-\partial_z\\
			\partial_x &\partial_y &\partial_z &0
		\end{bmatrix},   \notag
	\end{align*} 
    where 
    $	\bar{\varepsilon} = \frac{\ln \varepsilon}{2}, \;
    \bar{\mu} = \frac{\ln \mu}{2}$.
		Following \cite{vaha2020,kh2022}, the Riemann-Silberstein vector \cite{Khan2005An}  is defined by
		\begin{equation}
			\BF^{\pm} = (\sqrt{\varepsilon} \BE \pm i \BB/\sqrt{\mu})/\sqrt{2}.
		\end{equation}
	   We define new variables and source term based on Riemann-Silberstein vector as 
	   \begin{align*}
	   \Psi &= \ket{0}\otimes \bm{\psi^+}+\ket{1}\otimes \bm{\psi^-},\quad  \FJ=\ket{0}\otimes J^+ +\ket{1}\otimes J^-,
	   \end{align*}
      with the vector $\bm{\psi^{\pm}}$ and $J^{\pm}$ defined by
      \begin{align*}
    \Psi^{\pm} &=\frac 12\bigg(
    (-F_x^{\pm} \pm iF_y^{\pm})\ket{0}
    +F_z^{\pm}\ket{1}+F_z^{\pm}\ket{2}
    +(F_x^{\pm} \pm iF_y^{\pm})\ket{3}
    \bigg),\\
    J^{\pm}&= \frac{1}{2\sqrt{2\varepsilon}}
      \bigg((-J_x \pm iJ_y)\ket{0}
        +(J_z+v\rho)\ket{1}+(J_z-v\rho)\ket{2}
        +(J_z\pm iJ_y)\ket{3}
      \bigg).
\end{align*}
  Then, we  write Maxwell's equation in a matrix form as
   \begin{align} \label{eq:maxwell matrix}
   	\frac{\partial \Psi}{\partial t} 
   	= (M_0 +M')\Psi - \FJ
   	=v\bigg(-\begin{bmatrix}
   		\bm{\Sigma} \cdot \bm{\nabla} &\textbf{0}\\
   		\textbf{0}  &\bm{\Sigma}^{*} \cdot \bm{\nabla}
   	\end{bmatrix}
   	+\frac 12 \begin{bmatrix}
   		M_{11}' &M_{12}'\\
   		M_{21}' &M_{22}'
   	\end{bmatrix}\bigg)\Psi -\FJ,
   \end{align}
   where $M_0+M'=(T \Cm T^{\dag})$,  and the unitary matrix $T$ and Hermitian matrix $\bm{\Sigma}$ are defined by
		\begin{equation}
			T = \frac 12 
			\begin{bmatrix}
				-1 &i &0 &0  &-i &-1 &0 &0 \\
				0 &0 &1 &i  &0  &0  &i &-1\\
				0 &0 &1 &-i &0  &0  &i &1 \\
				1 &i &0 &0  &i  &-1 &0 &0\\
				-1 &-i&0 &0  &i  &-1 &0 &0\\
				0 &0 &1 &-i &0  &0  &-i&-1\\
				0 &0 &1 &i  &0  &0  &-i&1\\
				1 &-i&0 &0  &-i &-1 &0 &0 
			\end{bmatrix},
		\quad 
			\bm{\Sigma} = 
		\begin{bmatrix}
			\bm{\sigma} &\textbf{0}\\
			\textbf{0} &\bm{\sigma}
		\end{bmatrix}=
		\textbf{1}\otimes \bm{\sigma}.
		\end{equation}
	Here $\bm{\Sigma} \cdot \bm{\nabla} = \Sigma_1 \partial_x +\Sigma_2 \partial_y+\Sigma_3\partial_z$, $\Sigma_i=\textbf{1} \otimes \sigma_i$, and
	the Pauli matrices are defined by 
	\begin{equation*}
			\sigma_1 = \begin{bmatrix}
			0 & 1\\
			1 & 0
		\end{bmatrix},
		\quad
		\sigma_2 = \begin{bmatrix}
			0 & -i\\
			i & 0
		\end{bmatrix}, \quad
		\sigma_3 = \begin{bmatrix}
			1 &0 \\
			0 &-1
		\end{bmatrix}.
	\end{equation*}
       The conjugate of  $\bm{\Sigma}$ is denoted by
		$\bm{\Sigma}^*$.
		Each component of $M'$ is defined by 
		\begin{align*}
			M_{11}'&= \bm{\sigma}\cdot (\bm{\nabla} (\bar{\varepsilon}+\bar{\mu}))\otimes \textbf{1},\quad\quad
			M_{12}' = \bm{\sigma}\cdot (\bm{\nabla} (\bar{\varepsilon}-\bar{\mu}))\sigma_2)\otimes \sigma_2,\\
			M_{21}'&=\bm{\sigma}^*\cdot (\bm{\nabla} (\bar{\varepsilon}-\bar{\mu}))\sigma_2)\otimes \sigma_2, \quad
			M_{22}' = \bm{\sigma}^*\cdot (\bm{\nabla} (\bar{\varepsilon}+\bar{\mu}))
			\otimes \textbf{1}.	
		\end{align*}
	 The electromagnetic field is recovered by a unitary matrix, i.e. 
	 $\Cf = T^{\dagger} \Psi$.

		  Mathematically, Equation~\eqref{eq:maxwell matrix} is equivalent to Equation~\eqref{eq:maxwell matrix 1} after applying a unitary transformation.
		In a source-free homogeneous medium, the matrix $\Cm$ has nonzero 4-dimensional matrix blocks only along the off-diagonal directions. However, $M_0$ has nonzero 2-dimensional matrix blocks along the diagonal, and all other entries of the matrix are zero, which is a direct sum of four Pauli matrix blocks. Since the structure of the time evolution equation resulting from the Riemann-Silberstein formulation is simpler to work with in the qubit framework, we simulate Equation~\eqref{eq:maxwell matrix} instead of Equation~\eqref{eq:maxwell matrix 1}. 
  
		\section{Quantum simulation via Schr$\ddot{\text{o}}$dingerisation}\label{sec:schr}
  
		In this section, we briefly review the Schr$\ddot{\text{o}}$dingerisation approach first proposed in \cite{JLY22a,JLY22b} for general linear ODEs, which is written as
		\begin{equation} \label{eq:ODE}
			\frac{\D}{\D t}\Bu = A(t) \Bu(t) +\Bb(t),\quad 
			\Bu(0) = \Bu_0,
		\end{equation}  
		where $\Bu$, $\Bb\in \bbC^n$ and $A\in \bbC^{n\times n}$. 
		It is noted that all semi-discrete systems (after spatial discretizations) of  (PDEs) are ODE systems.
		Using an auxiliary scalar function $r(t)\equiv 1$, the above ODEs can be rewritten as a homogeneous system
		\begin{equation}
			\frac{\D }{\D t} \begin{bmatrix}
				\Bu \\
				r
			\end{bmatrix}=
			\begin{bmatrix}
				A &\Bb\\
				\textbf{0}^{\top} &0
			\end{bmatrix}\begin{bmatrix}
				\Bu \\
				r
			\end{bmatrix},\qquad
			\begin{bmatrix}
				\Bu(0)\\ r(0)
			\end{bmatrix} = \begin{bmatrix}
				\Bu_0\\ 1
			\end{bmatrix}.
		\end{equation}
		Therefore, without generality, we assume $\Bb =\textbf{0}$ in \eqref{eq:ODE}. 
		Since any matrix can be decomposed into a Hermitian term and an anti-Hermitian one, Equation \eqref{eq:ODE} can be expressed as
		\begin{equation}\label{eq:ODE1}
			\frac{\D}{\D t} \Bu = H_1 \Bu +iH_2 \Bu, \quad  \Bu(0) = \Bu_0,
		\end{equation}
		with $H_1 = \frac{A+A^{\dag}}{2}$ and $H_2 = \frac{A-A^{\dag}}{2i} $, both Hermitian. 
		Using the warped phase transformation $\Bw(t,p) = e^{-p}\Bu$ for $p>0$ and symmetrically extending the initial data to $p<0$, Equation \eqref{eq:ODE1} is converted to a system of linear convection equations:
		\begin{equation}\label{eq:up}
			\frac{\D}{\D t} \Bw = -H_1 \partial_p \Bw + iH_2\Bw, \quad 
			\Bw(0) = e^{-|p|}\Bu_0.
		\end{equation}
		The solution $\Bu$ can be restored by 
		\begin{equation}\label{eq: recover uh e}
			\Bu = e^{p^*}\Bw, \quad \text{for} \; \text{some} \; p^*>0,
		\end{equation}
	    or using the integration to obtain
		\begin{equation}\label{eq: recover uh integration}
			\Bu = \int_0^{\infty} \Bw \;dp.
		\end{equation}
		To discretize the $p$ domain, we choose a large enough domain $p\in [L,R]$ (so wave initially supported inside the domain remains so in the duration of computation) and set the uniform mesh size  $\triangle p = (R-L)/N$ where $N$ is a positive even integer and grid points denoted by $L=p_0<\cdots<p_{N}=R$. Define the vector $\Bv$  the collection of the function $\Bw$ at these grid points by
		\begin{equation}
			\Bv= \sum_{j\in [n]} \ket{j} \otimes \Bv_j,\quad \Bv_j =\sum_{k\in [N]} \Bw_j(t,p_k) \ket{k},
		\end{equation}
	    where $\Bw_j$ is the $j$-th component of $\Bw$.
		The $1$-D basis functions for the Fourier spectral method are usually chosen as
		\begin{equation} \label{eq:phi nu}
			\phi_l^p(x) = e^{i\nu_l^p x},\qquad \nu_l^p = 2\pi (l-N/2-1)/(R-L),\quad 1\leq l\leq N. 
		\end{equation}
		Using \eqref{eq:phi nu}, we define 
		\begin{equation}
			\Phi^p = (\phi_{jl}^p)_{N\times N} = (\phi_l^p(p_j))_{N\times N},  \quad
			D_{p} = \text{diag}\{\nu_1^p,\cdots,\nu_{N}^p\}.
		\end{equation}
		Considering the Fourier spectral discretisation on $p$, one easily gets
		\begin{equation}
			\frac{\D}{\D t} \Bv = -i(H_1\otimes P) \Bv + i(H_2\otimes \textbf{1}_{N})\Bv.
		\end{equation} 
		Here $P$ is the matrix representation of the momentum operator $-i\partial_p$  and
		defined by $P = \Phi^p D_{p} (\Phi^p)^{-1}$.
		By a change of variables $\tilde{\Bv}=[\textbf{1}_n\otimes (\Phi^p)^{-1}]\Bv$, one gets
		\begin{equation}\label{eq:}
			\frac{\D}{\D t} \tilde{\Bv} = -i[(H_1 \otimes D_{p}) - (H_2\otimes \textbf{1}_N)] \tilde{\Bv} =-iH\tilde{\Bv}.
		\end{equation}
		At this point, a quantum simulation can be carried out to the above Hamiltonian system. For time-dependent Hamiltonians, we refer to \cite{AFL21,AFL22,BCSWW22,FLT23} for quantum algorithms. 
		In practice, $H_1$ and $H_2$ are usually sparse,  hence the Hamiltonian $H=H_1\otimes D_{p}-H_2\otimes \textbf{1}_M$ inherits the sparsity.
		It is easy to find that 
		\begin{equation}
			s(H)= \mathscr{O}(s(A)), \quad
			\|H\|_{\max} \leq \|H_1\|_{\max}/\triangle p +\|H_2\|_{\max},
		\end{equation}
		where $s(H)$ is the sparsity of the matrix $H$ (maximum number of nonzero entries in each row) and $\|H\|_{\max}$ is its max-norm (value of largest entry in absolute value). In quantum algorithms, Hamiltonian simulation with nearly optimal dependence on  all parameters can be found in \cite{BCK15}, with complexity given by the next lemma.
		\begin{lemma}\label{lem:hamiltonian complexity}
			An $s$-sparse Hamiltonian H action on $m_H$ qubits can be simulated within error $\delta$ with
			\begin{equation}
				\mathscr{O}\bigg(\tau \log(\tau/\delta)/(\log\log(\tau/\delta))\bigg)
			\end{equation} 
		queries and 
		\begin{equation}
			\mathscr{O}\bigg(
		    \tau \big[m_H+\log_{2.5}(\tau/\delta)\big]
		    \frac{\log(\tau/\delta)}{\log\log(\tau/\delta)}	
			\bigg)
		\end{equation}
    	additional 2-qubits gates, where $\tau = s\|H\|_{max} t$, and $t$ is the evolution time.
			\end{lemma} 

   After the computation of $\tilde{\Bv}$, one can use quantum (inverse) FFT to get back to $\Bv$, and then $\Bu$ via \eqref{eq: recover uh e} or \eqref{eq: recover uh integration}.  For more details on the choice of $p^*$ in \eqref{eq: recover uh e}  or numerical integration of \eqref{eq: recover uh integration}, we refer to Section~2 in \cite{JLY22b}.
   
		\section{Quantum simulation of Maxwell's equations in a linear  homogeneous medium with periodic boundary conditions } \label{sec:schr periodic BD}
		
		In this section, we first consider quantum simulation for Maxwell's equations with periodic boundary conditions in a linear homogeneous medium, namely
		$\varepsilon$ and $\mu$ are constants. In this case, the matrix $M'$ in \eqref{eq:maxwell matrix} disappears. 
		 We choose a uniform spatial mesh size $\triangle x  = \triangle y=\triangle z= M^{-1}$ for $M $ with  an even positive integer. 
   
		\subsection{Quantum Simulation of \eqref{eq:maxwell matrix} with the spectral method} \label{sec:periodic matrix homogeneous}
  
		The $3$-dimensional grid points are given by $\Bx_{\Bj}=(x_{j1},y_{j2},z_{j3})$ with $\Bj = (j_1,j_2,j_3)$, and 
		\begin{equation}
			x_{j_1} = j_1 \triangle x,  \quad 
			y_{j_2} = j_2 \triangle y, \quad 
			z_{j_3} = j_3\triangle z,\quad j_1,j_2,j_3\in [M].
		\end{equation}
		Let the $n_{\Bj}^i$-th component of the vector $\bm{\psi}_h$ that approximates $\psi_{i}(t,x_{\Bj})$ be denoted by $\bm{\psi}_{i,\Bj}$, where
		\begin{equation}
			n_{\Bj}^i = M^{3}i+\sum_{k=1}^3 j_k M^{k-1}, \quad \ket{\Bj} =\ket{j_3} \otimes \ket{j_2} \otimes \ket{j_1},\quad i\in [8],
		\end{equation}
        and $\psi_i$ is the $i$-th component of $\Psi$.
		Therefore, one has
		\begin{equation}\label{eq:psih}
			\bm{\psi}_h =\sum_{i \in [8]} \ket{i } \otimes (\sum_{\Bj} \bm{\psi}_{i,\Bj}\ket{\Bj}).
		\end{equation}
		The discretization for the source term $\FJ$ is denoted by
		\begin{align}\label{eq:Jh}
			\BJ_h = \sum_{i\in [8]} \ket{i} \otimes (\sum_{\Bj} \FJ_{i}(t,x_{\Bj})\ket{\Bj}) .\quad
		\end{align}
		The 1-D basis functions for the Fourier spectral method in $x$-space are defined by
		\begin{equation*}
			\Phi = (e^{i\nu_l x_j})_{M\times M}, \quad D_{\nu} = \text{diag}\{\nu_1,\cdots,\nu_M\},\quad 
			\nu_l = 2\pi(l-M/2-1), \quad 1\leq l\leq M.
		\end{equation*}
		Considering the Fourier spectral discretization on $\Bx$,  one easily gets
		\begin{equation}\label{eq: varphi}
           \begin{aligned}
			\frac{\D}{\D t} \bm{\psi}_h &= \bbQ \bm{\psi}_h - \BJ_h,\quad 
         \bbQ = -iv \begin{bmatrix}
				\bbP_1  &\textbf{0}\\
				\textbf{0}   &\bbP_1^{*}
			\end{bmatrix},\quad
			\bbP_1 =\sum_{i=1}^3\Sigma_i \otimes\BP_i .
           \end{aligned}
		\end{equation}
		The matrices $\BP_1$, $\BP_2$ and $\BP_3$ are defined by 
		\begin{equation}
			\BP_1 = \bm{\Phi} \BD_{1} \bm{\Phi}^{-1},\quad
			\BP_2 = \bm{\Phi} \BD_{2} \bm{\Phi}^{-1},\quad
			\BP_3 = \bm{\Phi} \BD_{3} \bm{\Phi}^{-1}.
		\end{equation}
		Here $\bm{\Phi} = \Phi^{\otimes^3}$,
		$\BD_{1} = {\textbf{1}_M}^{\otimes^2}\otimes D_{\nu}$, 
		$\BD_{2} = \textbf{1}_M \otimes D_{\nu} \otimes \textbf{1}_M$, 
		$\BD_{3} = D_{\nu} \otimes \textbf{1}_M^{\otimes^2}$.
        Note that the Fourier spectral discretization will generate a matrix  which is not sparse, so this may affect the complexity used in Lemma 3.1 and sparse access.
		Thus, let $\bm{c}(t) = (\textbf{1}^{\otimes^3}\otimes\bm{\Phi}^{-1} )\bm{\psi}_h$, $\tilde{\BJ} = (\textbf{1}^{\otimes^3}\otimes\bm{\Phi}^{-1} ) \BJ_h$, Equation \eqref{eq: varphi} is rewritten as
		\begin{equation} \label{eq:maxwell simple}
  \begin{aligned}
			\frac{\D}{\D t} \bm{c} =\tilde{\bbQ}\bm{c} -\tilde{\BJ}, \quad
			\tilde{\bbQ} = -iv \begin{bmatrix}
				\bbQ_1 &\textbf{0}\\
				\textbf{0} &\bbQ_1^*
			\end{bmatrix},\quad
			\bbQ_1 = \sum_{i=1}^3 \Sigma_i \otimes \BD_i.
   \end{aligned}
		\end{equation}
		It is obvious to see that $\tilde{\bbQ}$ is anti-Hermitian.
		We rewrite Equation~\eqref{eq:maxwell simple} in a homogeneous form
		\begin{equation}
			\frac{\D}{\D t} \Bu = A\Bu, \quad
			\Bu = \begin{bmatrix}
				\Bc(t) \\
				r(t)
			\end{bmatrix},\quad
			A = \begin{bmatrix}
				\tilde\bbQ  &-\tilde{\BJ}\,\\
				\textbf{0}^{\top} &0
			\end{bmatrix}, \quad
			\Bu(0) =  \begin{bmatrix}
				\Bc(0) \\
				1
			\end{bmatrix},
		\end{equation}
	which is a $n=8M^3+1$ dimensional ODE system.
		With the help of the preceding calculation,
		we are now in a position to apply Schr$\ddot{\text{o}}$dingerisation.
		In terms of \eqref{eq:maxwell simple}, 
		using a new variable
		\begin{equation}
			\tilde{\Bv} = (\textbf{1}_{n}\otimes\Phi^{-1})\sum_{i,j} \ket{n_{\Bj}^i}\otimes \Bv_{n_{\Bj}^i},
			\quad
			\Bv_{n_{\Bj}^i} = \sum_k \psi_{i,j}(t,p_k)\ket{k},
		\end{equation}
		one gets an ODE system that suits a quantum simulation:
   \begin{equation}\label{eq:hamiltion psi}
			\frac{\D}{\D t} \tilde{\Bv} = -i[(H_1\otimes D_{p})-(H_2\otimes \textbf{1}_{N})]\tilde{\Bv}=-iH \tilde{\Bv},
   \end{equation}
		where the matrices $H_1$ and $H_2$ are defined by 
		\begin{equation}
			H_1 = \frac 12\begin{bmatrix}
				\textbf{0} & -\bm{\tilde{\BJ}}\\
				-\bm{\tilde{\BJ}}^{\top} & 0 
			\end{bmatrix},\quad
			H_2 = \frac{1}{2i} \begin{bmatrix}
				2\tilde \bbQ & -\bm{\tilde{\BJ}} \\
				\bm{\tilde{\BJ}}^{\top} &0
			\end{bmatrix}.
		\end{equation}
	   \begin{theorem}\label{thm: complexity 1}
	   	Given sparse-access to the Hermitian matrix $H$ in \eqref{eq:hamiltion psi} and the unitary $U_{\text{initial}}$ that prepares the initial quantum state $\ket{\Bu(0)}=U_{\text{initial}} \ket{0}$. Assume the mesh size satisfies $N=\mathscr{O}(M)$ and $M=2^m$. With the 
	Schr$\ddot{\text{o}}$dingerisation method, the state $\ket{\Bu(t)}$ can be simulated with gate complexity given by
	   	\begin{equation}
	   		N_{Gate}= \mathscr{O}\big(M((d+2)m^2+4m)/\log m\big)
          +\mathscr{O}(m\log m),
	   	\end{equation}
       where $d$ is the dimensional number.
	   \end{theorem}
    \begin{proof}
    Given the initial state $\ket{\Bu(0)}$, one gets the following procedure
    \begin{equation*}
     \Bu(0)\xrightarrow{DFT}\Tilde{\Bv}(0)\xrightarrow{e^{-iHt}}\Bv(t)\xrightarrow{DFT} \Bu(t).
    \end{equation*}
        It is known that the quantum Fourier transforms in one dimension can be implemented using $\mathscr{O}(m \log m)$ gates. 
        Under the  assumption of the mesh size, the lack of regularity of the initial condition  implies first-order accuracy on the spatial discretization, the error bound $\delta$ satisfies 
        $M^{-1}\sim N^{-1}\sim \delta$.
        Considering $s(H)=\mathscr{O}(1)$, $t=\mathscr{O}(1)$ and $\|H\|_{\max}=\max\{\|H_1\|_{\max},\|H_2\|_{\max}\}=\mathscr{O}(M)$, one has
        \begin{equation*}
            \frac{\log \tau \delta^{-1} }{ \log \log \tau \delta^{-1}}= \mathscr{O}(\frac{m}{\log m}).
        \end{equation*}
    The proof is finished by 
        Lemma~\ref{lem:hamiltonian complexity} .
    \end{proof}

  It is well known that the complexity of the FDTD  is $\mathscr{O}(M^{d+1})$ under the given error bound $\delta\sim M^{-1}$ 
  \cite{GL1996},  which is much larger than quantum algorithms. However, Schr$\ddot{\text{o}}$dingerisation with  Yee's algorithm for quantum simulation  not only reduces the complexity, but also retains some advantages of the FDTD schemes.

		\subsection{Quantum simulation of \eqref{eq:maxwell} with Yee's algorithm}\label{sec:Quantum simulation with Yee algorithm}
  
		Yee's finite difference method \cite{TafYee00} for Maxwell equations is the most popular algorithm 
		for numerically approximating Maxwell's equations, due to its simplicity and preservation of the continuous vector identities on the discrete grid.
		In this section, we use Yee's lattice discretization of the spatial operator in Equation \eqref{eq:maxwell}.  
		Without causing any ambiguity, let $\BE $ denote $\sqrt{\varepsilon}\BE$, $\BB$ denote $\BB/\sqrt{\mu}$. The Maxwell equations in a linear homogeneous medium is written as
		\begin{subequations}\label{eq:maxwell Eq EB}
			\begin{align}
				\frac{\partial }{\partial t} \BE -v \nabla \times \BB &= -\frac{\BJ}{\sqrt{\varepsilon }},\quad
				\frac{\partial }{\partial t} \BB +v \nabla \times \BE = \textbf{0},\\
				\nabla \cdot \BB &=0, \quad \quad \quad 
				\nabla \cdot \BE = \frac{\rho}{\sqrt{\varepsilon}}. \label{eq:maxwell Gauss law, maxwell Thomson}
				\end{align}
		\end{subequations}  
	   From Equation~\eqref{eq:maxwell Eq EB}, the Maxwell-Gauss equation (or Gauss's law) and the Maxwell-Thomson equation~\eqref{eq:maxwell Gauss law, maxwell Thomson} are 
	   actually consequences of the other equations and charge conservation equation
	   \begin{equation}
	   	\frac{\partial \rho}{\partial t}  = \nabla \cdot \BJ.
	   \end{equation} 
	                        
		The different components of the electromagnetic field
		and of the current densities are calculated at the cell center (half integer index) and at the cell vertices (integer index) according to  Yee's lattice configuration:
		\begin{align*}
			\BE_{\Bj} &=(E_{x,\Bj},E_{y,\Bj},E_{z,\Bj})
			=(E_{x,j_1,j_2+\frac 12,j_3+\frac 12},\;
			E_{y,j_1+\frac 12,j_2,j_3+\frac 12},\;
			E_{z,j_1+\frac 12,j_2+\frac 12,j_3}),\\
			\BB_{\Bj} &=(B_{x,\Bj},B_{y,\Bj},B_{z,\Bj})=
			(B_{x,j_1+\frac 12,j_2,j_3},\;
			B_{y,j_1,j_2+\frac 12,j_3},\;
			B_{z,j_1,j_2,j_3+\frac 12}).
		\end{align*}
		Correspondingly, the current densities are calculated at the cell center (half integer index) and  the cell vertices (integer index) according to  Yee's lattice configuration:
		\begin{equation*}
			\BJ_{\Bj} =(J_{x,\Bj},J_{y,\Bj},J_{z,\Bj})/\sqrt{\varepsilon}
			=(J_{x,j_1,j_2+\frac 12,j_3+\frac 12},\;
			J_{y,j_1+\frac 12,j_2,j_3+\frac 12},\;
			J_{z,j_1+\frac 12,j_2+\frac 12,j_3}).
		\end{equation*}
		Following Yee's algorithm, one gets the semi-discrete system
		\begin{align}
				\frac{\D \BE_{h}}{\D t} 
				-v \nabla_h \times \BB_{h} &= -\BJ_{h}, \label{eq:semi-Eh}\\
				\frac{\D \BB_{h}}{\D t}  +v\nabla_h \times \BE_{h}  &= \textbf{0}, \label{eq:semi-Bh}
			\end{align}
			where $\BE_h$, $\BB_{h}$ and $\BJ_h$ are the collections of $\BE_{\Bj}$, $\BB_{\Bj}$ and $\BJ_{\Bj}$, 
			\begin{equation}
				\BE_h = \begin{bmatrix}
					\sum_{\Bj} E_{x,\Bj}\ket{\Bj}\\
					\sum_{\Bj} E_{y,\Bj}\ket{\Bj}\\
					\sum_{\Bj} E_{z,\Bj}\ket{\Bj}
				\end{bmatrix}, \quad
				\BB_h = \begin{bmatrix}
					\sum_{\Bj} B_{x,\Bj}\ket{\Bj}\\
					\sum_{\Bj} B_{y,\Bj}\ket{\Bj}\\
					\sum_{\Bj} B_{z,\Bj}\ket{\Bj}
				\end{bmatrix}, \quad
				\BJ_h = \begin{bmatrix}
					\sum_{\Bj} J_{x,\Bj}\ket{\Bj}\\
					\sum_{\Bj} J_{y,\Bj}\ket{\Bj}\\
					\sum_{\Bj} J_{z,\Bj}\ket{\Bj} 
				\end{bmatrix}.
			\end{equation}
			The discrete curl operator $\nabla_h \times$ is the central difference, which  is also used in divergence operators.
			Define a matrix
			\begin{equation}\label{eq:def FM}
				\BF_{M}=\sum_{i=0}^{M-2}\ket{i}\bra{i+1}+\ket{M-1}\bra{0}, \quad i\in [M].
			\end{equation}
			Using \eqref{eq:def FM}, we define the following matrices
			\begin{align*}
                \BC_{x} = \frac{\textbf{1}_{M^2}\otimes \BF_M -\textbf{1}_{M^3}}{\triangle x},\;
                \BC_{y} = \frac{\textbf{1}_M\otimes \BF_M\otimes \textbf{1}_M-\textbf{1}_{M^3}}{\triangle y},\;
                \BC_{z} = \frac{\BF_M \otimes\textbf{1}_{M^2} -\textbf{1}_{M^3}}{\triangle z}.
			\end{align*}
			The matrix expression for \eqref{eq:semi-Eh}-\eqref{eq:semi-Bh} is rewritten as
			a $n=6M^3+1$ dimensional ODE system:
			\begin{align}\label{eq:matrix FDTD}
				\frac{\D }{\D t} \Bu = A \Bu,\quad
				\Bu = 
				\begin{bmatrix}
					\BE_h\\
				    \BB_h\\
				    r
				\end{bmatrix}, \quad 
			A = 
				\begin{bmatrix}
					\textbf{0} &\BM_{B}^E  &-\BJ_h \\
					\BM_{E}^{B} &\textbf{0}  &\textbf{0}\\
					\textbf{0}  &\textbf{0} &0
				\end{bmatrix} ,\quad 
			\Bu(0) 
		=\begin{bmatrix}
			\BE_h(0)\\
			\BB_h(0)\\
			1
		\end{bmatrix},
			\end{align}
			where $\BM_{B}^E,\BM_{E}^B \in \bbR^{3M^3\times 3M^3}$, and they satisfy $\BM_{E}^B=-(\BM_{B}^E)^{\top}$, 
			\begin{equation}
              \BM_{B}^E = v \begin{bmatrix}
              	\textbf{0}   &-\BC_z        &\BC_y\\
              	\BC_z  		 &\textbf{0}    &-\BC_x\\
              	-\BC_y  	 &\BC_x  		&\textbf{0}
               \end{bmatrix},
			\end{equation}
		   where the zero vector $\textbf{0}$ has the same size as $\BC_x$. 
		   Comparing with \eqref{eq:matrix FDTD} and \eqref{eq:semi-Eh}-\eqref{eq:semi-Bh}, it can be seen that $\BM_{B}^E$ is the matrix expression
		   of the discrete curl operator.

			Applying the Schr$\ddot{\text{o}}$dingerisation, one gets an Hamiltonian system for the new variable $\Bv = (\textbf{1}_{n}\otimes \Phi^{-1})\Bu$, 
			\begin{equation}\label{eq:hamiltonian FDTD}
				\frac{\D}{\D t} \Bv(t) = -i (H_1\otimes D_{p}) \Bv + i (H_2\otimes \textbf{1}_N) \Bv = -i H\Bv,
			\end{equation}
			where the matrices $H_1$ and $H_2$ are defined by
			\begin{equation}\label{eq:def H1 H2}
				H_1 = \begin{bmatrix}
					\textbf{0} &\textbf{0} & -\BJ_h/2\\
					\textbf{0} &\textbf{0} & \textbf{0}\\
					-\BJ_h^{\top}/2   &\textbf{0} &0
				\end{bmatrix},\quad
				H_2 = \frac{1}{i}\begin{bmatrix}
					\textbf{0} &\BM_{B}^E  &-\BJ_h/2\\
					\BM_{E}^B  &\textbf{0} &\textbf{0}\\
					\BJ_h^{\top}/2    &\textbf{0} &0
				\end{bmatrix}.
			\end{equation}
		 \begin{theorem}
			Given sparse-access to the Hermitian matrix $H$ in \eqref{eq:hamiltonian FDTD}, assume the same assumptions in Theorem~\ref{thm: complexity 1} hold true. The 
			Schr$\ddot{\text{o}}$dingerisation method can be simulated with gate complexity given by
			\begin{equation}
				N_{Gate}=N_{Gate}= \mathscr{O}\big(M((d+2)m^2+2m)/\log m\big)
          +\mathscr{O}(m\log m).
			\end{equation}
		\end{theorem}
	   \begin{proof}
	   	The proof is the same as that in Theorem~\ref{thm: complexity 1},
	   	and we omit it here.
	   	\end{proof}
      
		 It is well known that  Yee's algorithm for the space derivatives satisfies
		\begin{equation}\label{eq:discrete curl and divergence}
			\nabla_h\cdot (\nabla_h\times \BE_h)=0,\quad
			\nabla_h \cdot (\BE_h\times \BB_h) =  
			\BB_h \cdot (\nabla_h \times \BE_h)-\BE_h\cdot (\nabla_h \times \BB_h).		
		\end{equation}
		The discrete divergence of \eqref{eq:semi-Bh} gives the solenoidal  property: 
		\begin{equation}
			\nabla_h \cdot \BB_h(t) = 0,\quad \forall t>0,
		\end{equation}
		as long as the initial magnetic field is divergence free.
		Similarly, the discrete Gauss's law is enforced once 
		the charge  conservation is ensured.  
		We  now check if the quantum simulation  with
		Yee's scheme  still has the property.
		It is easy to find that Equation~\eqref{eq:hamiltonian FDTD} is the discretization of 
		Equation~\eqref{eq:up} with $H_1$ and $H_2$ defined in \eqref{eq:def H1 H2}, and one has 
	$$\Bw = e^{-p}\begin{bmatrix}
		\BE_h\\
		\BB_h \\
		r
	\end{bmatrix}, \quad 
    \Bw(0) = e^{-|p|}\begin{bmatrix}
    \BE_h\\
    \BB_h \\
    1
\end{bmatrix}.
$$
Taking the discrete divergence of \eqref{eq:up} and using \eqref{eq:discrete curl and divergence}, one gets
\begin{align}
	\frac{\D}{\D t} (	e^{-p} \nabla_h \cdot \BE_h)
   =\,&(\nabla_h\cdot  \frac{\BJ_h}{2}) \,\big[\partial_p (e^{-p} r) - e^{-p} r\big], 
   \label{eq:divergence Eh}
   \\
   \frac{\D}{\D t} (	e^{-p} \nabla_h \cdot \BB_h)
   =\,& \textbf{0}. \label{eq:divergence Bh}
\end{align}
From \eqref{eq:divergence Bh}, it can be seen that the discrete magnetic field is divergence-free if we recover $\BB_h$ by Equation \eqref{eq: recover uh integration}. 
Since the discretization of $[\partial_p (e^{-p} r) - e^{-p} r]$ in the $p$-domain dose not equal to $-2 e^{-p} r$ due to the error from the spectral method and  the non-smoothness of the initial value, the discrete Gauss law may not be exactly ensured. However, it is satisfied within  the numerical tolerance error.  This is because the overall error comes from two parts, one from the spatial discretization of Yee's algorithm and the spectral method, and the other from the evolution of the Hamiltonian system. In particular, when the source of the system vanishes, both the discrete Gauss law and the divergence free magnetic field are preserved due to the disappearance of $H_1$. Moreover, the quantum algorithm preserves the energy conservation law when the current density is set to zero.

We  now compare the simulation between the discretization of  \eqref{eq:maxwell matrix} and \eqref{eq:maxwell Eq EB}. Obviously, 
the  computational complexity of the latter is smaller. 
Let $\Cf_i$ denote the $i$-th component of $\Cf$, from \eqref{eq:maxwell matrix 1}, one obtains
\begin{equation}
	\partial_t \Cf_4 = v \nabla \cdot \BB/\sqrt{2 \mu},\quad 
	\partial_t \Cf_8 =- v\nabla \cdot \sqrt{\varepsilon/2} \BE + v \rho/\sqrt{2 \varepsilon},
\end{equation}
in a linear homogeneous medium. It is known that $\Cf_4 \equiv 0$, $\Cf_8\equiv 0$,  discrete Gauss law and Maxwell-Thomson equation hold modulus error from discretization of $\Cf_4$ and $\Cf_8$.
Besides, different from quantum simulation \eqref{eq:hamiltonian FDTD}, the discrete energy conservation cannot be guaranteed when the current density disappears.

	\section{Quantum simulation of Maxwell's equations in a linear  homogeneous medium with physical boundary conditions}
    \label{sec:schr physical boundary conditions}
    
				In this section, we concentrate on how physical boundary conditions can be incorporated into the framework of  Schr$\ddot{\text{o}}$dingerisation, especially for the boundary conditions mentioned above in Section~\ref{sec:A brief review of Maxwell equations}.
			For simplicity, we assume $\varepsilon$ and $\mu$ are constants.
   
			\subsection{Quantum simulation of \eqref{eq:maxwell matrix} with the upwind scheme}
			\label{sec:Quantum simulation BD}
   
			According to the perfect conductor conditions in \eqref{eq:EB perfect conductor}, we consider the corresponding boundary conditions of the variable $\Psi$. Firstly, we give a matrix representation of the boundary conditions to the variable $\Cf$:
			\begin{equation*}
				\begin{bmatrix}
					\BN_1  &\textbf{0}\\
					\textbf{0} &\BN_{2}
				\end{bmatrix} \Cf = \Cb_{pc} \Cf =\textbf{0},\;
			\BN_1 = 
             \begin{bmatrix}
			0     &\frac{n_z}{\sqrt{\varepsilon}}    &\frac{-n_y}{\sqrt{\varepsilon}} &0  \\
				\frac{-n_z}{\sqrt{\varepsilon}}   &0    &\frac{n_x}{\sqrt{\varepsilon}}  &0  \\
				\frac{n_y}{\sqrt{\varepsilon}}    &\frac{-n_x}{\sqrt{\varepsilon}} &0    &0  \\
				0 &0  &0 &1 \\
			\end{bmatrix},\;
		\BN_2 = \begin{bmatrix}
				 0 &0 &0 &0\\
		         0 &0 &0 &0\\
		         0 &0 &0 &0\\
		         0 &0 &0 &1\\
		\end{bmatrix},
			\end{equation*}
			where the unit normal vector to the boundary is denoted by $\Bn=(n_x,n_y,n_z)^{\top}$, and we have used the fact that 
			$\Cf_4=\Cf_8=0$ on $\partial \Omega$.
			Since ${\Psi} = T \Cf$, one gets the boundary condition for ${\Psi}$,
			\begin{align}
				\Cb_{pc}\Cf &= \Cb_{pc}T^{\dagger}\bm{\Psi} = \bbB_{pc}\bm{\Psi}=
    \begin{bmatrix}
					\BB_{pc}^1  &(\BB_{pc}^1)^*\\
					\BB_{pc}^2  &\BB_{pc}^2 
				\end{bmatrix} \bm{\Psi}=
    \textbf{0}.
			\end{align}
			Here the matrices  $\BB_{pc}^1,\, \BB_{pc}^2 \in \bbC^{4\times 4}$ are denoted by
			\begin{equation*}
			\BB_{pc}^1 = \frac{1}{\sqrt{\varepsilon}}\begin{bmatrix}
				-in_z  &-n_y &-n_y,  &-in_z        \\
				n_z    &n_x  &n_x,   &-n_z         \\
				- n_y + in_x  & 0   &0      &n_y + in_x    \\
				  0 &-i\sqrt{\varepsilon} &i\sqrt{\varepsilon} &0\\
			\end{bmatrix},\quad
		\BB_{pc}^2 = \begin{bmatrix}
	     0 &0 &0 &0\\
		 0 &0 &0 &0 \\
		0 &0 &0 &0\\
		 0 &-1 &1 &0
		\end{bmatrix}.
				\end{equation*}
		According to \eqref{eq:EB impedance boundary condition}, the impedance boundary condition is written with the unknowns $\BE$ and $\BB$ as
			\begin{equation}
				\BE\times \Bn + v \Bn \times (\BB\times \Bn) = 0, \quad \text{on}\; \partial \Omega.
			\end{equation}
			Similarly, the matrix representation of the impedance boundary condition corresponding to the variable $\Psi$ is 
			\begin{equation}
				\bbB_{im}{\Psi} =\begin{bmatrix}
					\BB_{im} & \BB_{im}^*
				\end{bmatrix}{\Psi}=\textbf{0}, \quad 
    \BB_{im} = \begin{bmatrix}
                    \BB_{im}^{L1} & \BB_{im}^{R1}\\
                    \BB_{im}^{L2} & \BB_{im}^{R2}
                \end{bmatrix},
			\end{equation}
			where the matrices $\BB_{im}^{Lj},\;\BB_{im}^{Rj}\in \bbC^{4\times 2}$, $j=1,2$ are defined by
            \begin{align*}
            &\BB_{im}^{L1}=\begin{bmatrix}
                n_x n_y \sqrt{\mu}v + i\sqrt{\mu}v(n_y^2 + n_z^2) - \frac{in_z}{\sqrt{\varepsilon}}  &\frac{-n_y}{\sqrt{\varepsilon}} + i\sqrt{\mu} v n_x n_z \\
				\frac{n_z}{\sqrt{\varepsilon}} - \sqrt{\mu}v(n_x^2 + n_z^2) - i\sqrt{\mu}vn_x n_y    & \frac{n_x}{\sqrt{\varepsilon}} + i\sqrt{\mu}vn_y n_z \\
				 n_y n_z \sqrt{\mu}v - \frac{n_y}{\sqrt{\varepsilon}} + i\frac{n_x}{\sqrt{\varepsilon}} -i\sqrt{\mu}v n_x n_z & -i\sqrt{\mu}v(n_x^2 + n_y^2)\\
					0         &-i       
                \end{bmatrix}, \;
                \BB_{im}^{L2}=\begin{bmatrix}
                    0         &0           \\
					0         &0       \\
					0         &0           \\
					0        &-1          
                \end{bmatrix},
            \\
            &\BB_{im}^{R1}=\begin{bmatrix}
                \frac{-n_y}{\sqrt{\varepsilon}} + i\sqrt{\mu} v n_x n_z
                &  n_x n_y \sqrt{\mu}v + i\sqrt{\mu}v(n_y^2 + n_z^2) - \frac{in_z}{\sqrt{\varepsilon}}  \\
				\frac{n_x}{\sqrt{\varepsilon}} +  i\sqrt{\mu} v n_y n_z
                & \frac{n_z}{\sqrt{\varepsilon}} - \sqrt{\mu}v(n_x^2 + n_z^2) - i\sqrt{\mu}vn_x n_y  \\
				 -i\sqrt{\mu}v(n_x^2 + n_y^2)
                & n_y n_z \sqrt{\mu}v - \frac{n_y}{\sqrt{\varepsilon}} + i\frac{n_x}{\sqrt{\varepsilon}} -i\sqrt{\mu}v n_x n_z \\
					i         &0       
                \end{bmatrix},\;
                \BB_{im}^{R2}=\begin{bmatrix}
                    0         &0           \\
					0         &0       \\
					0         &0           \\
					1        &0          
                \end{bmatrix}.  
            \end{align*}
			
			In the following, we only consider the 1-D case for the $z$-transverse electric (TE) model in the domain $[0,1]$.  For the three-dimensional case, a similar approach can be adopted straightforwardly.
			The electric field is assumed to have a longitudinal component $E_x$, and a transverse component $E_y$, i.e. $\BE=(E_x(x,t),E_y(x,t),0)$. 
			The magnetic field is aligned with the $z$ direction and its magnitude is denoted by $B_z$, i.e. $\BB=(0,0,B_z(x,t))$. 
			The reduced Maxwell system is written as
			\begin{equation}
				\frac{\partial E_x}{\partial t}= -J_x,\quad
				\frac{\partial E_y}{\partial t} +\frac{\partial B_z}{\partial x}= -J_y,\quad
				\frac{\partial B_z}{\partial t} +\frac{\partial E_y}{\partial x} = 0. \label{eq:1DEB}
			\end{equation}			
			The unit outward normal vector to the left of the computational domain is $[-1,0,0]^{\top}$, and to the right is $[1,0,0]^{\top}$.
			Without loss of generality, the boundary condition on the left is perfect conductor and the right is  impedance boundary condition.
			It is noted that $\Sigma_1$ can be diagonalized via a unitary matrix, that is 
			$\Sigma_1 = U^{\top} \Lambda_1 U $, with 
			\begin{equation}
				\Lambda_1 = \textbf{1} \otimes \begin{bmatrix}
					1 &0\\
					0 &-1
				\end{bmatrix},\quad 
				U = \textbf{1} \otimes \frac{1}{\sqrt{2}}
				\begin{bmatrix}
					1 & 1\\
					1 &-1
				\end{bmatrix}. 
			\end{equation}
			We rewrite the system of the variables $\tilde \Psi =(\textbf{1}\otimes U) \Psi$ as
			\begin{equation} \label{eq:tilde psi}
				\frac{\D }{\D t} \tilde \Psi  = -v \bm{\Lambda_1} \partial_x \tilde \Psi - \tilde \FJ ,
			\end{equation}
			where $\bm{\Lambda_1} = \textbf{1}\otimes \Lambda_1$ and $\tilde \FJ = (\textbf{1}\otimes U)\FJ$.
			The boundary condition of the perfect conductor on the left side of the domain is 
			computed by 
			\begin{equation}\label{eq:pc BD Psi}
				\bbB_{pc} \Psi = \bbB_{pc} (\textbf{1} \otimes U^{\top}) \tilde \Psi 
				= \tilde \bbB_{pc}  \tilde \Psi  = \textbf{0}.
			\end{equation}
			Simple calculation gives
			\begin{equation}
				\begin{bmatrix}
					\tilde\psi_0\\
					\tilde\psi_2\\
					\tilde\psi_4\\
					\tilde\psi_6
				\end{bmatrix}
				= \frac {1}{2\sqrt{2}}\begin{bmatrix}
					1 &1 &1 &-1\\
					-1 &-1 &1 &-1\\
					1 &-1 &1 &1\\
					1 &-1 &-1 &-1
				\end{bmatrix}
				\begin{bmatrix}
					\tilde\psi_1\\
					\tilde\psi_3\\
					\tilde\psi_5\\
					\tilde\psi_7
				\end{bmatrix}
				=B_{E2O} 	\begin{bmatrix}
					\tilde\psi_1\\
					\tilde\psi_3\\
					\tilde\psi_5\\
					\tilde\psi_7
				\end{bmatrix},
			\end{equation}
			where $\tilde \psi_i$ is the $i$-th value of $\tilde \Psi$, $i\in [8]$. Similarly, we have the impedance boundary condition on the right side of the domain:
			\begin{equation}
				\begin{bmatrix}
					\tilde\psi_1\\
					\tilde\psi_3\\
					\tilde\psi_5\\
					\tilde\psi_7
				\end{bmatrix}
				= \frac {1}{2\sqrt{2}}\begin{bmatrix}
					1 &-1 &-\frac{v-1}{v+1} &-\frac{v-1}{v+1}\\
					1 &-1 & \frac{v-1}{v+1} & \frac{v-1}{v+1}\\
					-\frac{v-1}{v+1} &-\frac{v-1}{v+1} &1 &-1\\
					\frac{v-1}{v+1} & \frac{v-1}{v+1} &1 &-1
				\end{bmatrix}
				\begin{bmatrix}
					\tilde\psi_0\\
					\tilde\psi_2\\
					\tilde\psi_4\\
					\tilde\psi_6
				\end{bmatrix}
				=B_{O2E}\begin{bmatrix}
					\tilde\psi_0\\
					\tilde\psi_2\\
					\tilde\psi_4\\
					\tilde\psi_6
				\end{bmatrix}.
			\end{equation} 
			We consider the upwind scheme for Equation \eqref{eq:tilde psi} on a uniform grid with space size $\triangle x = M^{-1}$.  The collection of the new variables $\tilde \Psi $ 
			and the source term are denoted by 
			\begin{equation}
				\Bu =\sum_{i\in [8]} \ket{i}\otimes 
				\big(\sum_{j_i \in [M]} \tilde{\psi}_{i,j_i}\ket{j_i}\big),
				\qquad
				\Bb = \sum_{i\in [8]} \ket{i} \otimes
				\big(\sum_{j_i\in [M]}\tilde \FJ_{i}(x_{j_i})\ket{j_i}\big),
			\end{equation}
			where $\tilde \psi_{i,j_i}$ is the approximation of $\tilde \psi_i(x_{j_i})$, and
			\begin{equation}
				x_{j_i} = (j_i+1) \triangle x, \quad i=0,2,4,6,\quad 
				x_{j_i} = j_i\triangle x, \quad i = 1,3, 5,7.
			\end{equation} 
			We define the finite difference operator $D_{x,L}$ when $i$ is odd, and $D_{x,R}$ if $i$ is even,
			\begin{equation*}
				D_{x,L} =  \frac{1}{\triangle x}(\textbf{1}_M- \sum_{i=1}^{M-1}\ket{i}\bra{i-1}), \quad 
				D_{x,R} = \frac{1}{\triangle x}(-\textbf{1}_M + \sum_{i=1}^{M-1} \ket{i-1}\bra{i}),
			\end{equation*}  
			It is easy to find  $D_{x,R} = -D_{x,L}^{\top}$. 
			Finally, one gets the
			system of  Equation~\eqref{eq:ODE} as 
			\begin{equation}\label{eq:AD IMPE}
				A =   \textbf{1}_4 \otimes \begin{bmatrix}
					-vD_{x,L} &0\\
					0 & vD_{x,R}  
				\end{bmatrix}
				+\frac{v}{\triangle x}
				\sum_{i,j=1}^4\big( (B_{E2O})_{i,j} \ket{O_{si}}\bra{E_{sj}} 
				+(B_{O2E})_{i,j} \ket{E_{ei}}\bra{O_{ej}}\big),
			\end{equation}
			where $(B_{{O2E}})_{i,j} = \bra{i}B_{O2E}\ket{j}$,  $(B_{{E2O}})_{i,j} = \bra{i}B_{E2O}\ket{j}$, and for $k\in [4]$,
			\begin{equation*}
				O_{sk} = 2Mk, \quad  O_{ek} = (2k+1)M-1,\quad
				E_{sk} = (2k+1)M,\quad  E_{ek} = 2(k+1)M-1.
			\end{equation*}

			\subsection{Quantum simulation of \eqref{eq:maxwell} with Yee's algorithm}
   
             Following the algorithm in Section~\ref{sec:Quantum simulation with Yee algorithm}, we handle the physical boundary condition with  Yee's algorithm.
            The discrete variables of the electromagnetic fields are
             \begin{equation}
             	\BE_j = (E_{x,j+\frac 12},E_{y,j},0),\quad 
             	\BB_j = (0,0,B_{z,j+\frac 12}).
             \end{equation}
			 The perfect conductor boundary condition on the left side of the domain for the discrete variables is obtained by
			 \begin{equation}
			 	E_{y,0} = 0.
			 \end{equation}
		    The impedance boundary condition on the right side of the domain for the discrete variables is 
		    \begin{equation}
		    	v (B_{z,M+\frac 12} + B_{z,M-\frac 12})/2- E_{y,M} = 0. 
		    	\end{equation}
			The collection of the discrete variables and the source term are denoted by
			\begin{equation}
				\Bu =\begin{bmatrix}
					\sum\limits_{j\in [M]} E_{x,j+\frac 12}\ket{j}\\
					\sum\limits_{j\in [M]} E_{y,j+1}\ket{j}\\
					\sum\limits_{j\in [M]} B_{z,j+\frac 12}\ket{j}
				\end{bmatrix}  ,\quad 
		 \Bb = \begin{bmatrix}
		 	-\sum\limits_{j\in [M]} J_{x}(x_{j+\frac 12})\ket{j}\\
		 	-\sum\limits_{j\in [M]} J_{y}(x_{j+1})\ket{j}\\
		 	\textbf{0}
		 	\end{bmatrix}.
			\end{equation}
			Using   Yee's algorithm,  one gets the system of Equation \eqref{eq:ODE} with the matrix $A\in \bbC^{3M\times 3M}$:
			\begin{equation}
			A = \begin{bmatrix}
				\textbf{0} &\textbf{0} &\textbf{0}\\
				\textbf{0} &\textbf{0} &-vD_{x,R} \\
				\textbf{0} &-vD_{x,L}     &\textbf{0}
				\end{bmatrix}
			 -\frac{v }{\triangle x} (\frac{2}{v}\ket{2M-1}\bra{2M-1}-\ket{2M-1}\bra{3M-1}).
			\end{equation} 
		Comparing with the algorithm in Section~\ref{sec:Quantum simulation BD}, it is easier to handle the physical boundary condition 
		for the quantum algorithm using  Yee's algorithm. 
  
  	\section{Quantum simulation of Maxwell's equations  for a linear  inhomogeneous medium} 
   \label{sec:schr inhomogeneous medium}
   
		In this section, the permittivity and permeability of  the medium may depend on the space--including even {\it discontinuous} functions. For simplicity, we may assume  the periodic boundary condition for Maxwell equations. The mesh for space discretisation is the same as 
			in Section~\ref{sec:periodic matrix homogeneous}.
			The collections of the parameters $\bar{\varepsilon}$, $\bar{\mu}$  and $v$ are defined by 
			\begin{equation}
			\bar{\bm{\varepsilon}} = \sum_{\Bj} \bar{\varepsilon}(t,x_{\Bj})\ket{\Bj},\quad
			\bar{\bm{\mu}} =\sum_{\Bj} \bar{\mu}(t,x_{\Bj})\ket{\Bj},\quad
			\bm{v} = \sum_{\Bj} v(t,x_{\Bj})\ket{\Bj}.
			\end{equation}
			Using the spectral method, one obtains the following ODE system for quantum simulation
\begin{equation} \label{eq:maxwell simple 2}
	\frac{\D}{\D t} \bm{\psi}_h =
	\bigg(-i \begin{bmatrix}
		\tilde{\bbQ}_{v1} &\textbf{0}\\
		\textbf{0} &\tilde{\bbQ}_{v1}
	\end{bmatrix}
+\frac 12\begin{bmatrix}
	\tilde{\bbV}_{v,11} &\tilde{\bbV}_{v,12}\\
	\tilde{\bbV}_{v,21} &\tilde{\bbV}_{v,22}
\end{bmatrix}\bigg)\bm{\psi}_h - \BJ_h, 
\end{equation}
where $\bm{\psi}_h$ and $\BJ_h$  are defined as in \eqref{eq:psih} and \eqref{eq:Jh}, and 
\begin{equation}\label{eq:Qv}
	\tilde{\bbQ}_{v1} = \Sigma_1 \otimes (\text{diag}\{\bm{v}\}\BP_1)
	+\Sigma_2 \otimes (\text{diag}\{\bm{v}\}\BP_2) +\Sigma_3 \otimes (\text{diag}\{\bm{v}\}\BP_3).
\end{equation}
The matrix $\tilde{\bbV}_{v,ij}$, $i,j=1,2$ are defined by
\begin{align}
	\tilde{\bbV}_{v,11} &= \sum_i \sigma_i \otimes \textbf{1}\otimes \tilde{V}_i^{+}, \quad \quad  \quad 
	\tilde{\bbV}_{v,12}=\sum_i (\sigma_i \sigma_2)\otimes \sigma_2 \otimes \tilde{V}_i^{-} \\
	\tilde{\bbV}_{v,21}&=\sum_i (\sigma_i^* \sigma_2) \otimes \sigma_2 \otimes \tilde{V}_i^{-}, \quad
	\tilde{\bbV}_{v,22}=\sum_i  \sigma_i^*\otimes \textbf{1} \otimes 
	\tilde{V}_{i}^{+},
\end{align}
and the matrices $	\tilde{V}_i^{\pm}$ $i=1,2,3$ are defined by 
\begin{align*}
	\tilde{V}_i^{\pm} =  \text{diag} \big\{\Bv \odot\big(\BP_i (\bar{\bm{\varepsilon}}{\pm}\bar{\bm{\mu}}) \big)\big\} ,
\end{align*}
where $\Bc=\Ba\odot\Bb$ such that $c_i=a_i b_i$, and $c_i$, $a_i$, $b_i$ are the $i$-th component of vector 
$\Bc$, $\Ba$ and $\Bb$, respectively.
Next we use the Schr$\ddot{\text{o}}$dingerisation approach to complete the quantum simulation.

Let us consider media with constant magnetic permeability $\mu=\mu_0$ but with scalar spatially varying permittivity $\varepsilon=\varepsilon(x)$.
For simplicity, let the electromagnetic wave propagate from 
a region of constant  permittivity $\varepsilon_1$ to a region of higher constant
permittivity $\varepsilon_1$, where 
\begin{equation}
    \varepsilon(x)  = \begin{cases}
        \varepsilon_1, \;x<L\\
        \varepsilon_2, \;x>L
    \end{cases}.
\end{equation}
Follow the idea in ~\cite{va20},  we approximate permittivity index profile  by the hyperbolic tagent function-profile
\begin{equation}
    \varepsilon(x) = \frac{\varepsilon_1+\varepsilon_2}{2}-\frac{\varepsilon_1-\varepsilon_2}{2}\tanh{\beta(x-L)},
\end{equation}
where $\beta$ controls the thickness of the boundary region between the two media. Periodic boundary conditions are enforced by adding a small buffer region after the end of the grid so that the refractive index is periodic (See Fig.~\ref{fig:dis_nx}). Finally, we apply the Schr$\ddot{\text{o}}$dingerisation method 
on \eqref{eq:maxwell simple 2}.
\begin{figure}
    \centering
		\subfigure[$\text{permittivity}$\; \text{index}\;]{
		\includegraphics[width=0.42\linewidth]{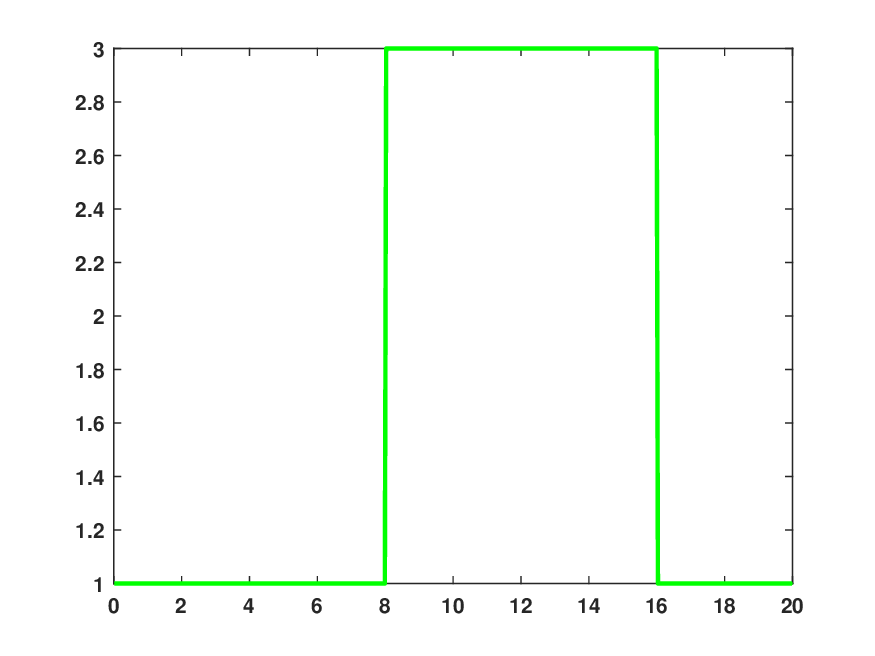}\label{fig:dis_para1}}
		\subfigure[$\text{approximate}\; $\text{permittivity}$\; \text{index}$]{
		\includegraphics[width=0.42\linewidth]{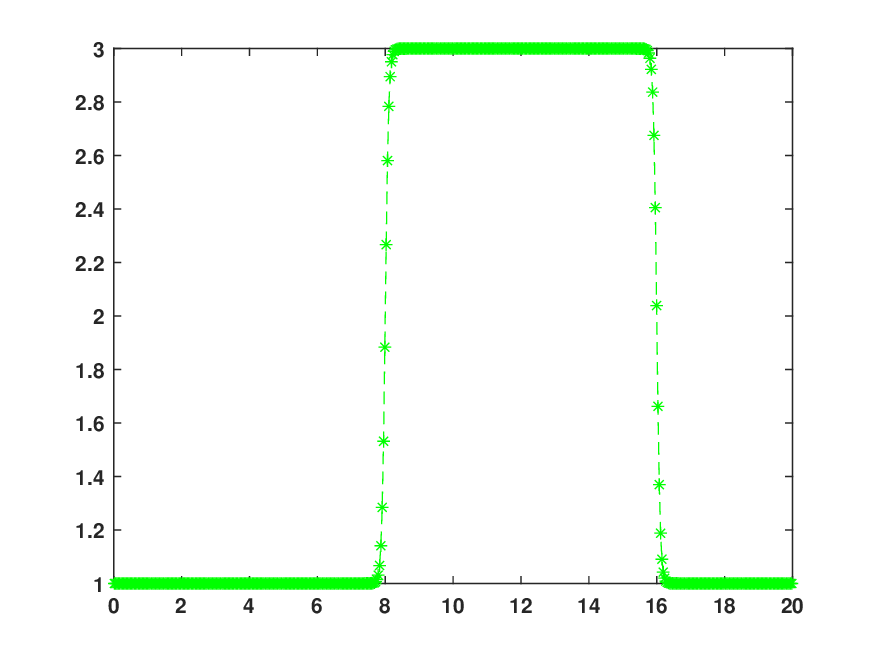}\label{fig:dis_para2}}
    \caption{Permittivity index}
    \label{fig:dis_nx}
\end{figure}

However, the matrix $\tilde \bbQ_v$ in \eqref{eq:maxwell simple 2} is not sparse, which may affect the complexity of quantum algorithms, and this approach may not simulate the complex media and physical boundary conditions efficiently. We instead consider the Schr$\ddot{\text{o}}$dingerisation combined with the immersed boundary method (or immersed interface method) .
For the complex media which has two dielectric material $\Omega^{1}$ with parameters $\varepsilon_1$, $\mu_1$ and $\Omega_2$  with parameters $\varepsilon_2$, $\mu_2$,  and satisfies $\Omega=\Omega_1\cup \Omega_2$,
We define 
the interface conditions for the electric and magnetic fields on the interface $\Gamma$ are given by
\begin{equation}\label{eq:interface codition}
	\hat{\Bn}\cdot [ \varepsilon \BE]=0,\quad 
	 \hat{\Bn}\cdot [\BB] =0, \quad
	\hat{\Bn} \times[ \BE]=0,\quad  \hat{\Bn} \times [\BB/\mu]=0, \quad \text{on}\,\Gamma,
\end{equation}
where $\hat{\Bn}$ is the unit normal of the interface pointing to $\Omega_1$ and  the jumps  on the interface are denoted by 
\begin{equation}
	[\BE] = \BE|_{\Omega_1} - \BE|_{\Omega_2},\quad
	[\BB] = \BB|_{\Omega_1} - \BB|_{\Omega_2},\quad \text{on}\,\Gamma.
\end{equation}
Combining with the interface condition for $\Cf_4$ and $\Cf_8$, i.e. 
\begin{equation}
	[\Cf_4]=0,\quad [\Cf_8]=0,
\end{equation}
the jump conditions in terms of the variable $\Cf$  are
\begin{equation}
	R_1 \Cf^1-R_2\Cf^2=\begin{bmatrix}
		R_1^{11} &\textbf{0}\\
		\textbf{0} & R_{1}^{22}
	\end{bmatrix}\Cf^1
-		\begin{bmatrix}
	R_2^{11} &\textbf{0}\\
	\textbf{0} & R_{2}^{22}
\end{bmatrix}\Cf^2 =\textbf{0},
\end{equation}
where the matrices $R_j^{11}$ and $R_j^{22}$, $j=1,2$ are defined by
\begin{equation*}
	R_j^{11} = \begin{bmatrix}
		\sqrt{\varepsilon_j} n_x &\sqrt{\varepsilon_j}n_y &\sqrt{\varepsilon_j} n_z &1\\
		0 &-\frac{n_z}{\sqrt{\varepsilon_j}} &\frac{n_y}{\sqrt{\varepsilon_j}} &1\\
		\frac{n_z}{\sqrt{\varepsilon_j}} &0 &-\frac{n_x}{\sqrt{\varepsilon_j}} &1\\
		-\frac{n_y}{\sqrt{\varepsilon_j}} &\frac{n_x}{\sqrt{\varepsilon_j}} &0 &1 \\
	\end{bmatrix},\quad
R_j^{22} = \begin{bmatrix}
	\sqrt{\mu_j} n_x &\sqrt{\mu_j}n_y &\sqrt{\mu_j} n_z &1\\
	0 &-\frac{n_z}{\sqrt{\mu_j}} &\frac{n_y}{\sqrt{\mu_j}} &1\\
	\frac{n_z}{\sqrt{\mu_j}} &0 &-\frac{n_x}{\sqrt{\mu_j}} &1\\
	-\frac{n_y}{\sqrt{\mu_j}} &\frac{n_x}{\sqrt{\mu_j}} &0 &1 \\
\end{bmatrix}.
	\end{equation*}
Using the transformation matrix $T$ and unitary matrix $U$, one gets the representation of interface conditions on $\Gamma$ in terms of $\tilde \Psi$ as
\begin{equation}\label{eq:tilde jump condition}
	\tilde R_1 \tilde{\Psi}^1 -\tilde R_2 \tilde \Psi^2 =0,
\end{equation}
where $\tilde R_j = R T^{\dagger} U^{\dagger}$.
An upwinding embedded boundary method \cite{CD03} can be applied for the wave equation~\eqref{eq:tilde psi} for the space  discretization, and then   
the Schr$\ddot{\text{o}}$dingerisation approach is used for the quantum simulation.

\section{Continuous-variable formulation} \label{sec:continuous-variable formulation}

We remark that the Schr\"odingerisation framework is not only applicable to qubit-systems, but also to {\it continuous-variable} 
 (CV) quantum systems. The continuous-variable analogue of a qubit is a qumode. The Schr\"odingerisation framework in the qumode representation is introduced in \cite{CVPDE2023}. 

Unlike a qubit, a  CV quantum state, or {\it `qumode'}, spans an infinite-dimensional Hilbert space. A qumode is the quantum analogue of a continuous classical degree of freedom. A qumode is acted upon by observables with a continuous spectrum, such as the position $\hat{x}$ and momentum $\hat{p}$ observables of a quantum particle. Its eigenbasis can be chosen to be for instance $\{|x\rangle\}_{x \in \mathbb{R}}$, which are the eigenstates of $\hat{x}$. It forms a complete basis so $\int |x\rangle \langle x|=I$. In this basis, a qumode can be expressed as, for instance, $|u(t)\rangle=(1/\|\vect{u}(t)\|)\int u(t,x)|x\rangle dx$, where $\|\vect{u}(t)\|^2=\int dx |u(t, x)|^2$ is the normalisation constant. A system of $m$-qumodes is a tensor product of $m$ qumodes. The qumode can also be acted upon by quadrature operators like the momentum $\hat{p}$ operator, where $[\hat{x}, \hat{p}]=i$. 

Maxwell's equations govern the dynamics of the electric $E_{x,y,z}$ and magnetic fields $B_{x,y,z}$, in the presence of charge density $\rho$ and current density $J_{x,y,z}$. These are all continuous quantities and it is interesting to ask if it is possible to represent this information in a continuous manner in a quantum device without first discretising. This could be made possible through a qumode representation. In addition,  it is easy to find from \eqref{eq:Qv} that  the matrix $\bbQ_v$ is not sparse due to the discretization of varying velocity.
However, this defect disappears in a continuous-variable quantum system.
It can be seen from \eqref{eq:maxwell simple 2}-\eqref{eq:Qv} that the advantages of the matrix representation based on Riemann-Silberstein vectors no longer exists for a linear inhomogeneous medium in our Schr\"odingerisation framework. Thus we simulate  Equation~\eqref{eq:maxwell matrix 1} instead of Equation~\eqref{eq:maxwell matrix} in this section.

In the qumode representation, vectors $\vect{F}$ in bold, e.g., $\vect{E}_{x,y,z}$, $\vect{B}_{x,y,z}$, $\vect{J}_{x,y,z}$ and $\vect{\rho}$ represents $\vect{F}=\iiint F(x,y,z)|x\rangle|y\rangle |z\rangle dx dy dz$, which is a quantum system consisting of 3 qumodes. We define
\begin{align*}
	& \vect{\mathcal{F}}=\sqrt{\varepsilon} \vect{E}_x |0\rangle+\sqrt{\varepsilon} \vect{E}_y |1\rangle+\sqrt{\varepsilon} \vect{E}_z |2\rangle+
    \frac{1}{\sqrt{\mu}} \vect{B}_x|4\rangle+\frac{1}{\sqrt{\mu}} \vect{B}_y|5\rangle+\frac{1}{\sqrt{\mu}} \vect{B}_z|6\rangle, \nonumber \\
	& \vect{\mathcal{J}}=\vect{J}_x|0\rangle+\vect{J}_y|1\rangle+\vect{J}_z|2\rangle-\vect{v}\vect{\rho}|8\rangle, \nonumber \\
	&	\mathcal{D}=\begin{pmatrix}
		0 & i\hat{p}_z & -i\hat{p}_y & i\hat{p}_x \\
		-i\hat{p}_z & 0 & i\hat{p}_x & i\hat{p}_y \\
		i\hat{p}_y & -i\hat{p}_x & 0 & i\hat{p}_z \\
		-i\hat{p}_x & -i\hat{p}_y & -i\hat{p}_z & 0
	\end{pmatrix}=-\hat{p}_z \otimes \textbf{1} \otimes \sigma_2-\hat{p}_x\otimes \sigma_2 \otimes \sigma_1+\hat{p}_y \otimes \sigma_2 \otimes \sigma_3,
\end{align*}
where $|0\rangle, |1\rangle,...,|8\rangle$ represent the $9$ possible states in the computational basis consisting of 3 qubits. In the continuous-variable framework, we can make the replacement $\partial_{x,y,z} \leftrightarrow -i\hat{p}_{x,y,z}$ where $\hat{p}$ is the momentum operator. In the continuous-variable framework, we can also make the replacement $x \leftarrow \hat{x}$ where $\hat{x}$ is the position operator obeying $[\hat{x},\hat{p}]=i$.
Then the following equation governing $\vect{\mathcal{F}}$ with an inhomogeneous term holds:
\begin{align*}
	 \frac{d\vect{\mathcal{F}}}{dt}&=-i\vect{A} \vect{\mathcal{F}}-\vect{\mathcal{J}} \nonumber \\
\vect{A}&=-i \sigma_2 \otimes \frac{v}{2} \mathcal{D}(2I-\bar{\varepsilon}-\bar{\mu})
	-\sigma_1\otimes\frac{v}{2}\mathcal{D}(\bar{\varepsilon}-\bar{\mu}),
\end{align*}
where $v=v(\hat{x},\hat{y},\hat{z})$, $\bar{\varepsilon} =\bar{\varepsilon}(\hat{x},\hat{y},\hat{z})$ and $\bar{\mu} = \bar{\mu}(\hat{x},\hat{y},\hat{z})$.
We define an operator $\hat{\eta}$ such that $\hat{\eta}\ket{\eta}=\eta \ket{\eta}$.
One can then apply our Schr\"odingerisation formulation to the inhomogeneous case. We Schr\"odingerise the system by dilating $\vect{\mathcal{F}} \rightarrow \vect{y}=\vect{\mathcal{F}} \otimes |0\rangle+\vect{\mathcal{J}}\otimes |1\rangle$ and  transform $\vect{y} \rightarrow \tilde{\vect{v}}$ (Schr\"odingerisation procedure) to obtain
the Hamiltonian matrix in \eqref{eq:}:
\begin{align*}
	\vect{H}=\vect{A}_2 \otimes \frac{1}{2}(\textbf{1}+\sigma_3) \otimes \hat{\eta}-\frac{\textbf{1}}{2}\otimes \sigma_1 \otimes \hat{\eta}+\vect{A}_1 \otimes \frac{1}{2}(\textbf{1}+\sigma_3) \otimes \textbf{1}+\frac{\textbf{1}}{2}\otimes \sigma_2 \otimes \textbf{1}=\vect{H}^{\dagger}, 
\end{align*}
where $\vect{A}_1=(\vect{A}+\vect{A}^{\dagger})/2=\vect{A}_1^{\dagger}$, $\vect{A}_2=i(\vect{A}-\vect{A}^{\dagger})/2=\vect{A}_2^{\dagger}$. 
This is quantum simulation on a system of 4 qumodes and 4 qubits.

\section{Numerical simulation} \label{sec:numerical simulation}

    For the numerical tests, we use the classical computer to simulate Hamilton system to validate the feasibility of the algorithms above. The solution for \eqref{eq:} at time $T$ is obtained by  
    \begin{equation*}\tilde{\Bv}=e^{-iH T}\tilde{\Bv}_0,\end{equation*}
    if $H$ is independent of $t$.  Otherwise, the backward Euler method is used to approximate the Hamiltonian system. First we use different quantum algorithms to simulate the propagation of the electromagnetic field without the source term, and then we consider the Schr$\ddot{\text{o}}$dingerisation method applied to physical boundary conditions. Finally, we simulate Maxwell's equations in media with material interfaces.
   
    \subsection{Periodic boundary conditions}
    
    In this test, we consider  $2D$ Maxwell's equations in the domain $\Omega=[0,2]^2$ for a $z$-transverse magnetic (TM) wave,
    where the magnetic field $\BB$ is transverse to the $z$-direction and electric field $\BE$
    has only one component along  the $z$-direction. 
    We set the cell number as $M=2^5$, $N=2^7$.
    The exact solution to the system is
    \begin{equation*}
    	E_z = \sin(\pi(x + 2y + {\sqrt{5}}t)),\quad
    	B_x = -2E_z/\sqrt{5},\quad
        B_y =  E_z/\sqrt{5}.
    	\end{equation*}
			 \begin{figure}[htbp]
			 	\centering
			 	\subfigure[$E_z$.]{
			 		\includegraphics[width=0.32\linewidth]{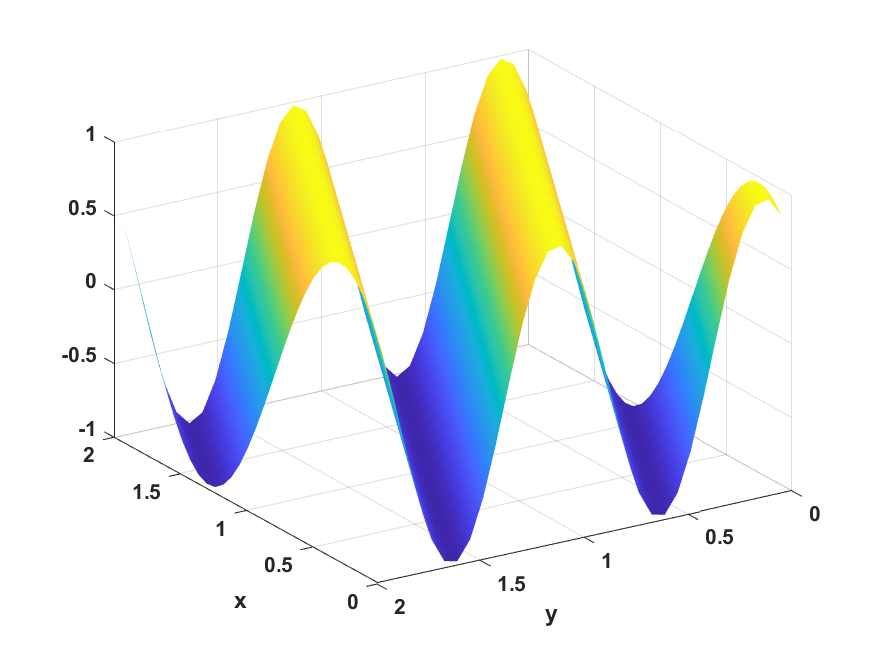}\label{fig:Ez_Ep_2D}}
			 	\subfigure[$E_z(x^*,y)$.]{
			 		\includegraphics[width=0.3\linewidth]{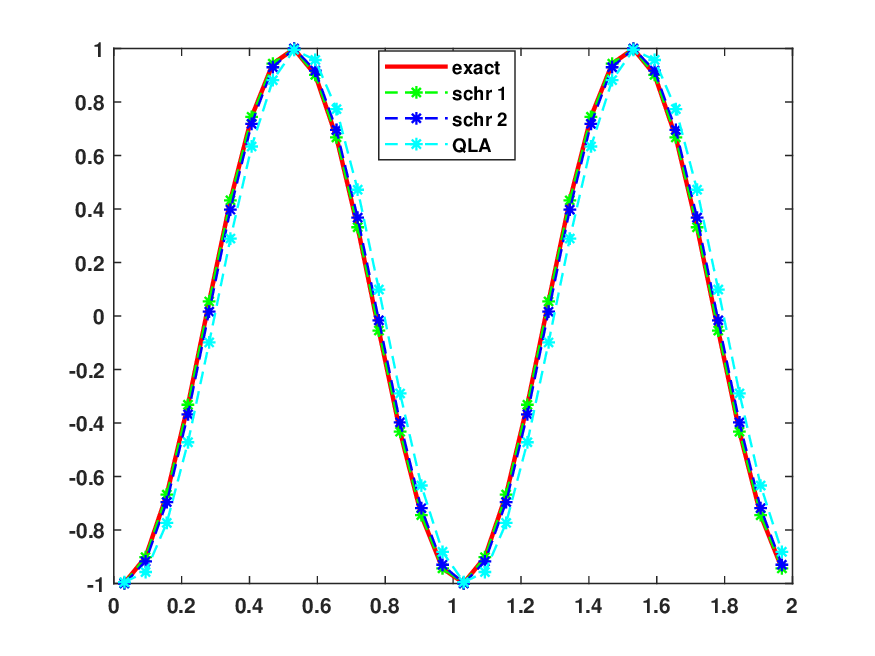}\label{fig:Ez_fixX}}
			 	\subfigure[$E_z(x,y^*)$]{
			 		\includegraphics[width=0.3\linewidth]{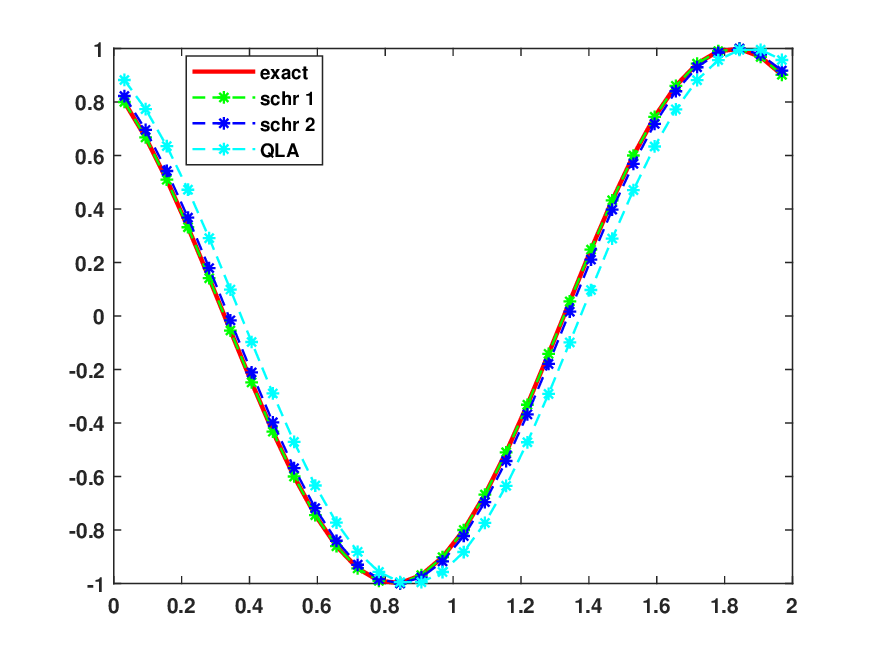}\label{fig:Ez_fixY}}
		 		\subfigure[$B_x$.]{
		 			\includegraphics[width=0.32\linewidth]{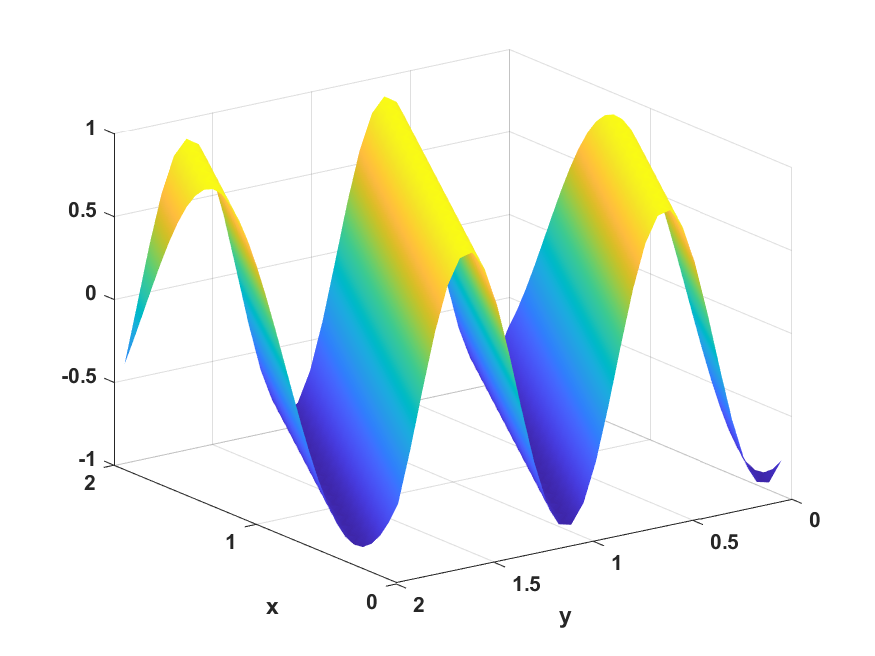}\label{fig:Bx_Ep_2D}}
		 		\subfigure[$B_x(x^*,y)$.]{
		 			\includegraphics[width=0.3\linewidth]{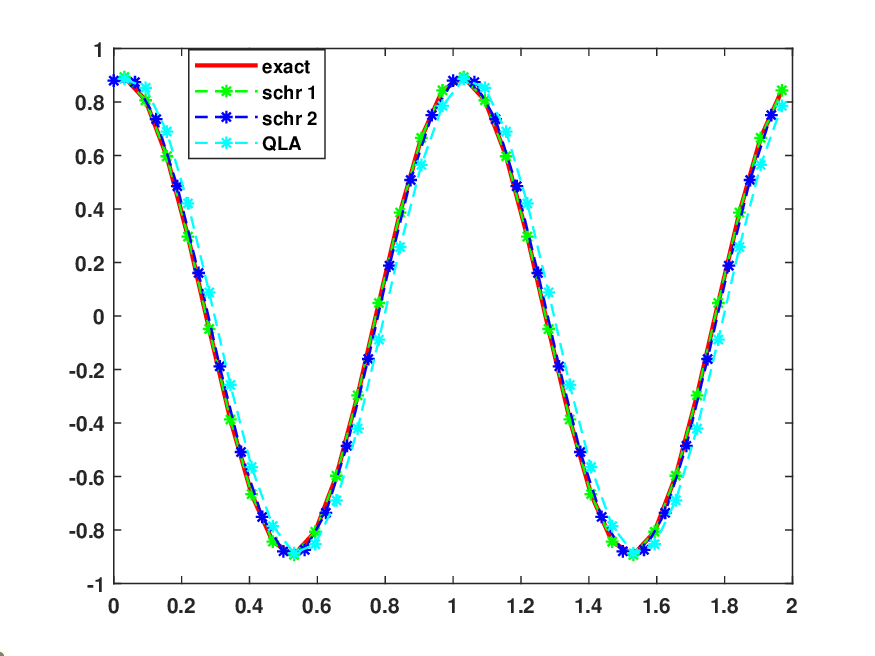}\label{fig:Bx_fixX}}
		 		\subfigure[$B_x(x,y^*)$]{
		 			\includegraphics[width=0.3\linewidth]{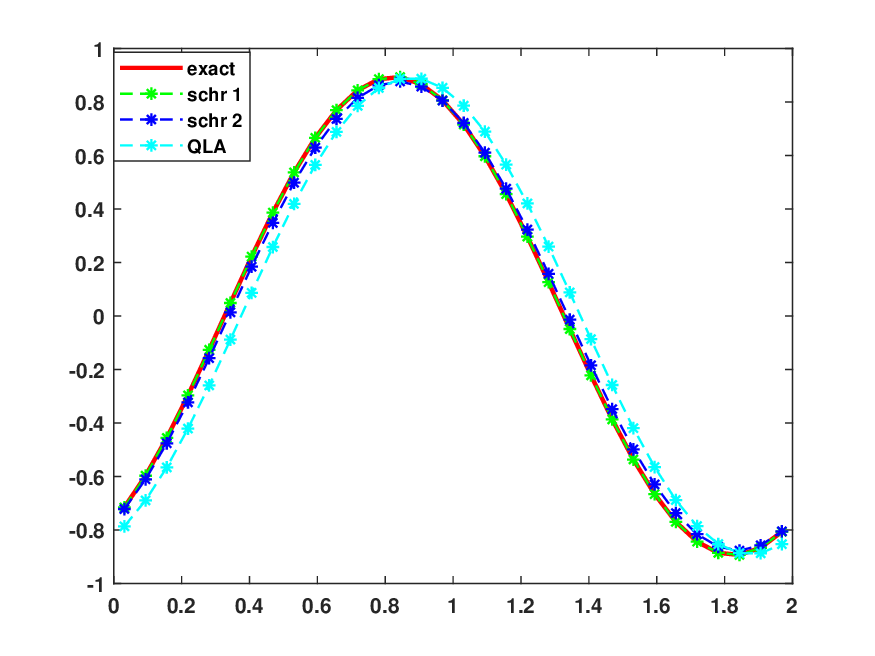}\label{fig:Bx_fixY}}
	 			\subfigure[$B_y$.]{
	 				\includegraphics[width=0.32\linewidth]{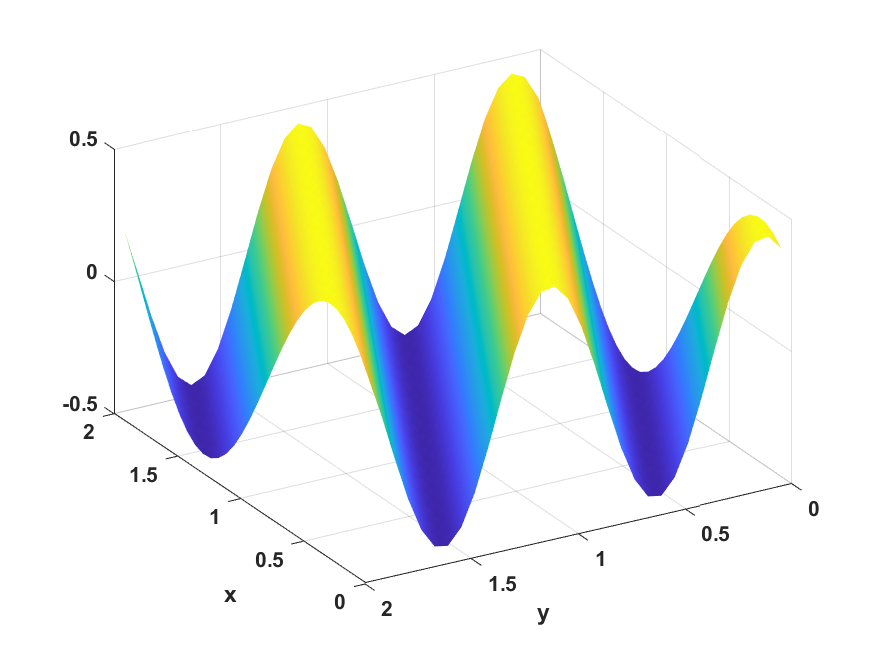}\label{fig:By_Ep_2D}}
	 			\subfigure[$B_y(x^*,y)$.]{
	 				\includegraphics[width=0.3\linewidth]{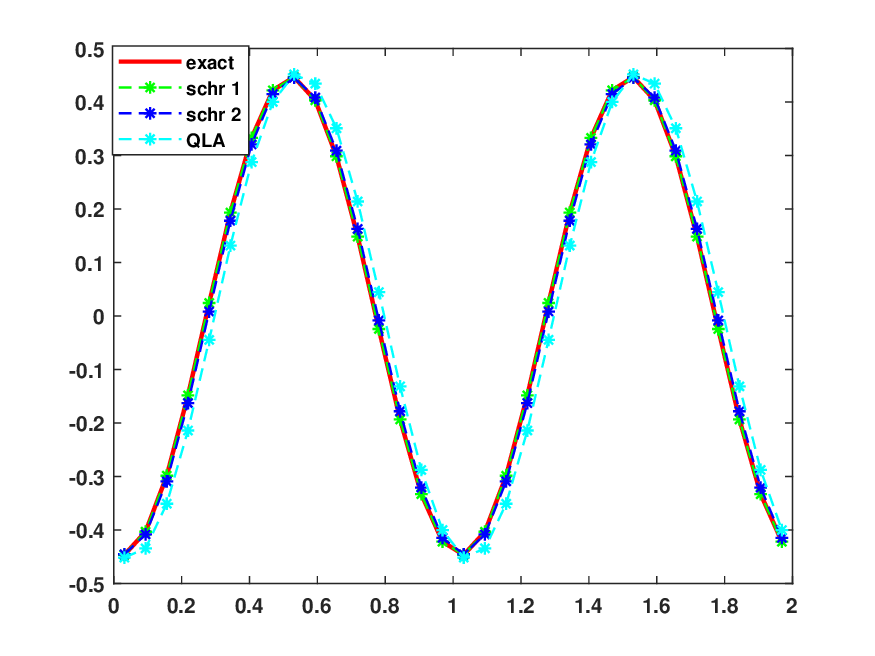}\label{fig:By_fixX}}
	 			\subfigure[$B_y(x,y^*)$]{
	 				\includegraphics[width=0.3\linewidth]{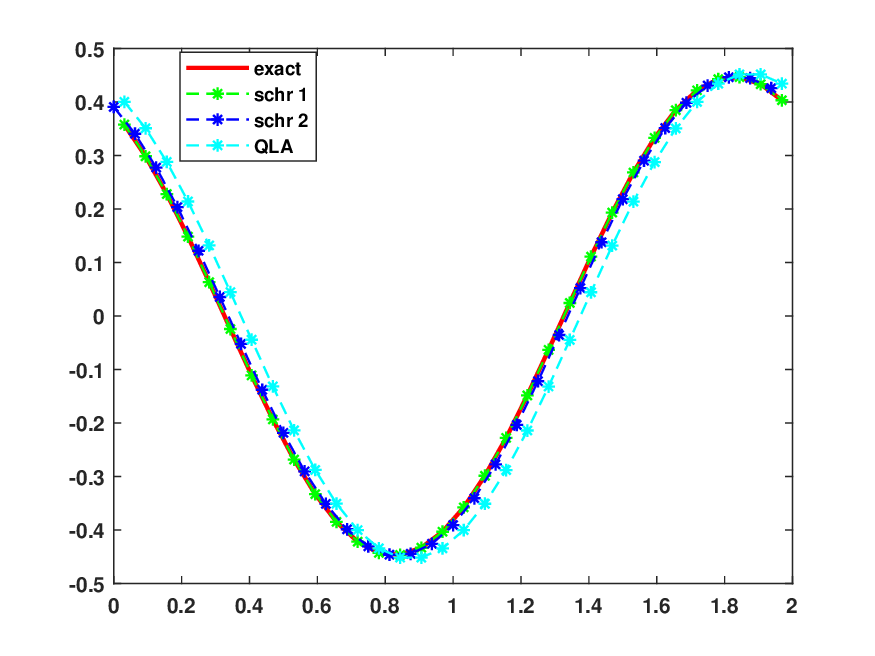}\label{fig:By_fixY}}
			 	\caption{Electromagnetic fields at $T=1$ with periodic boundary condition.
			 	On the left are the computed solutions with schr1, in the middle and right
		 	    are the computed and the exact solution with
               $x^*=39/32$ and $y^*=7/32$.}\label{fig:periodic}
			 \end{figure}
	
    \begin{table}[htp]
		  \centering
		 \begin{tabular}{cccccc}
		  \toprule [2pt]
		        & $\triangle \Ce$ &$\triangle (\nabla_h\cdot \BB_h)$ &$\Cf_{4,h}$  &$\Cf_{8,h}$ &$\text{err}_{EB}$\\
		  \midrule [2pt]
    QLA         &1.16e-4           & -       &3.71e-3   &3.72e-3   &1.53e-1\\
    schr 1      & 1.33e-15         & -       &9.72e-16 &9.70e-16 &3.72e-15\\
    schr 2      &4.44e-16          &6.88e-14 &-        &-        &3.83e-2\\
		  \bottomrule [2pt]
	     \end{tabular}
        \caption{Comparation of different schemes.}\label{tab}
		 \end{table}
	   Simulations of Equation~\eqref{eq:maxwell matrix}  are denoted by schr 1, 
     and  model \eqref{eq:maxwell} using the Yee's algorithm is called schr 2.
 	In Fig.~\ref{fig:periodic}, the comparison of $\text{schr 1}$, schr 2 and QLA proposed in \cite{va20,va202} shows that they are very close to the exact solution, and schr 1 ends up superior. 
    Define the discrete energy as 
    \begin{equation*}
        \Ce(T) =\sum_{\Bj} (E_{z,\Bj}^2+B_{x,\Bj}^2+B_{y,\Bj}^2)\triangle x^2.
    \end{equation*}
    From Table \ref{tab}, we find that the values of $\triangle \Ce =|\Ce(T)-\Ce(0)|$ and $\triangle (\nabla_h \cdot \BB_h)=|\nabla_h\cdot  \BB_h(t) - \nabla_h \cdot \BB_h(0)|$  are close to zero for schr 2 even when  $\text{err}_{EB}$ is much bigger, 
    where $\text{err}_{EB}$ is the error between the numerical and exact electromagnetic fields denoted by 
    \begin{equation*}
    \text{err}_{EB}=\max_{\Bj} \big(|E_{z,\Bj}-E_z(\Bx_{\Bj},T)|\; |B_{x,\Bj}-B_x(\Bx_{\Bj},T)|\; |B_{y,\Bj}-B_y(\Bx_{\Bj},T)|\big).
    \end{equation*}
    Let $\Cf_{4,h}$ denote the approximation of $\Cf_4$, which is the error of the Maxwell-Thomson equation. We find $\|\Cf_{4,h}\|_{l^{\infty}} \sim  \text{err}_{EB}$ for schr 1 and QLA.
    
		 \subsection{Physical boundary conditions with source terms}
			For simplicity of the exposition, we restrict ourselves to a reduced version of the Maxwell equations with one spatial variable, $x$,  namely 
			Equation~\eqref{eq:1DEB}. Moreover we assume $\sqrt{\varepsilon} = \sqrt{\mu} =1$.
			Let us consider a 1D case in $[0,L]$ with $L=15$, and use the exact solution to test the accuracy of algorithm.
			The cell number is set by $M =2^6$, $N= 2^{7}$.
			The simulation stops at $T=1$.	
	        The perfect conductor boundary condition for the TE model reads as
				\begin{equation}
					E_y(0)=0,\quad E_y(L)=0 .
				\end{equation}
			   In order to enforce the boundary condition, we set the exact solution as
			   \begin{equation}
			   	E_x = \sin(2\pi(x+t)/5),\quad 
               E_y = \frac{2\pi}{5}(\cos(2\pi x/5)/(2\pi)-1),\quad
			   	B_z = t\sin (2\pi x/5).
			   \end{equation}
		   Using the same method in \eqref{eq:pc BD Psi}, we get the perfect conductor boundary 
		   condition on the right side of the domain for Equation~\eqref{eq:tilde psi}
		   \begin{equation}\label{eq:PE left}
		   	\begin{bmatrix}
		   		\tilde \psi_1\\
		   		\tilde \psi_3\\
		   		\tilde \psi_5\\
		   		\tilde \psi_7
		   	\end{bmatrix}
	   	= \frac{1}{2 \sqrt{2}}\begin{bmatrix}
	   		1 &-1  &1  &1\\
	   		1 &-1 &-1 &-1\\
	   		1 & 1  &1 &-1\\
	   		-1 &-1 & 1 &-1
	   	\end{bmatrix}
   	\begin{bmatrix}
   		\tilde \psi_0\\
   		\tilde \psi_2\\
   		\tilde \psi_4\\
   		\tilde \psi_6
   	\end{bmatrix}
   =B_{O2E}\begin{bmatrix}
   	\tilde \psi_1\\
   	\tilde \psi_3\\
   	\tilde \psi_5\\
   	\tilde \psi_7
   \end{bmatrix}.
		   \end{equation}
			Changing the matrix $B_{O2E}$ in \eqref{eq:AD IMPE}, one gets the system for the Hamiltonian simulation with the perfect conductor boundary of the domain.
			The results are shown in  Fig.~\ref{fig:PEC_FDTD}.
			 \begin{figure}[htbp]
			\centering
			\subfigure[$E_x$]{
				\includegraphics[width=0.309\linewidth]{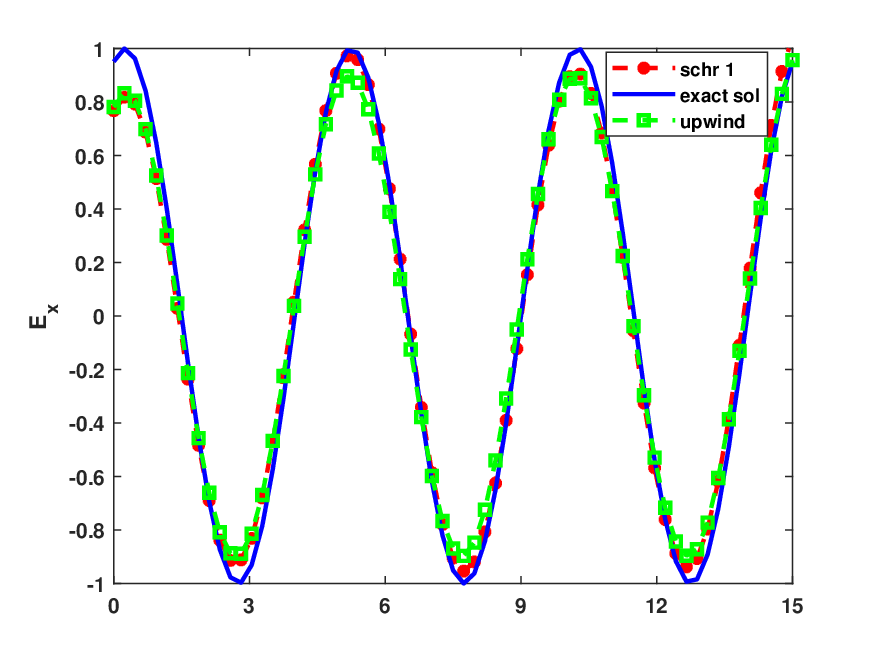}\label{fig:Ex_pec_upwind}}
			\subfigure[$E_y$]{
				\includegraphics[width=0.309\linewidth]{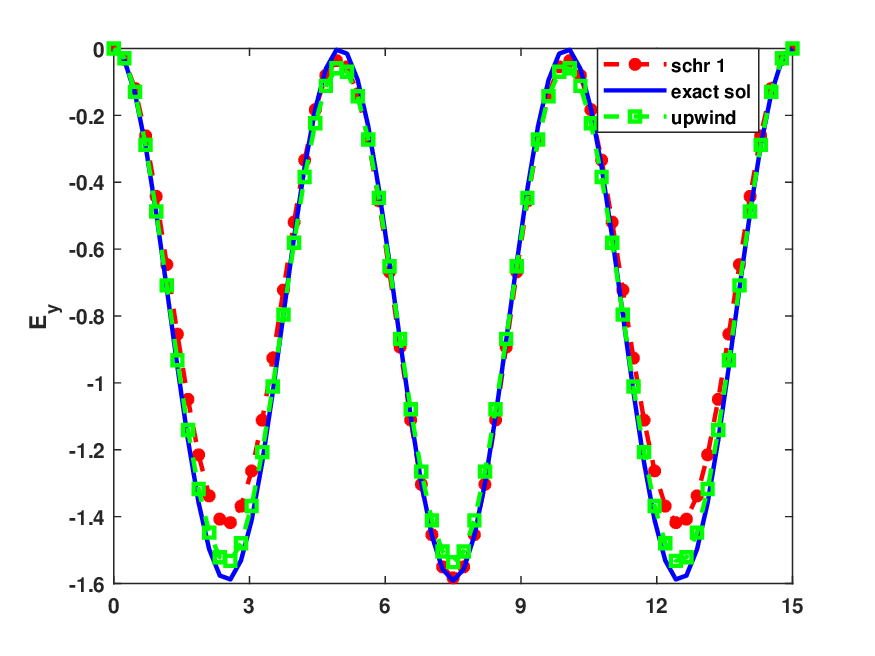}\label{fig:Ey_pec_upwind}}
			\subfigure[$B_z$]{
				\includegraphics[width=0.309\linewidth]{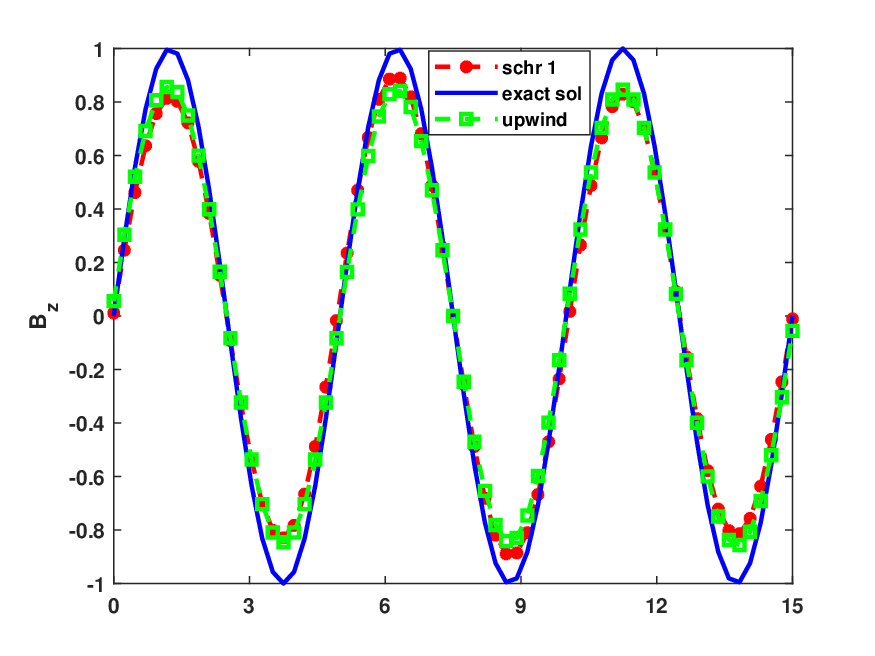}\label{fig:Bz_pec_upwind}}
			\subfigure[$E_x$]{
				\includegraphics[width=0.309\linewidth]{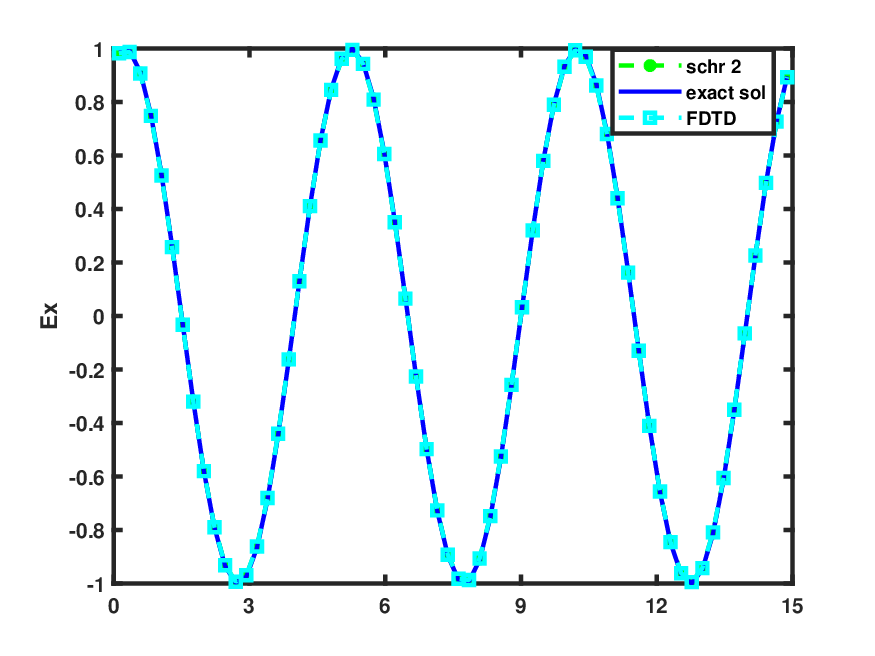}\label{fig:Ex_pec_FDTD}}
			\subfigure[$E_y$]{
				\includegraphics[width=0.309\linewidth]{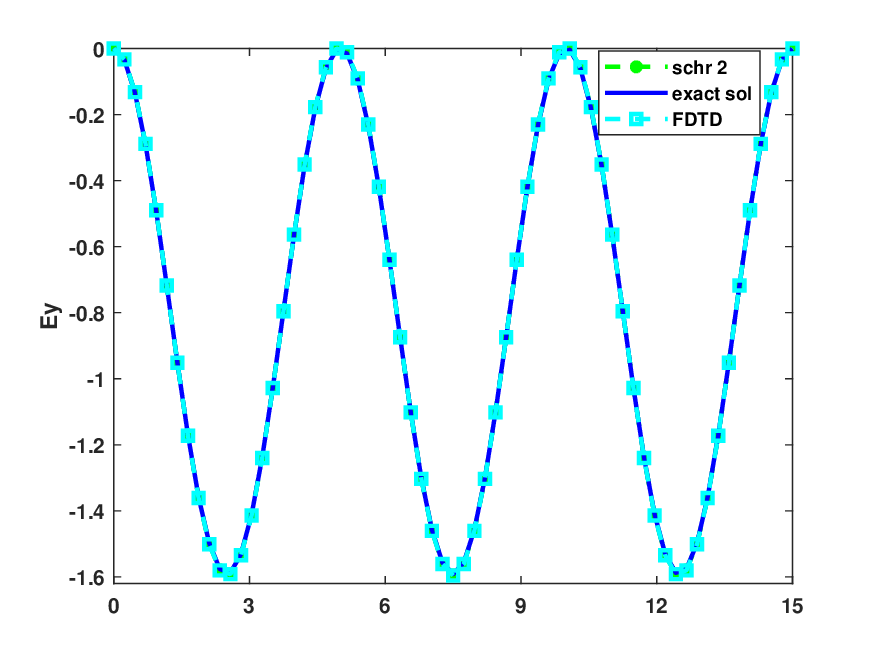}\label{fig:Ey_pec_FDTD}}
			\subfigure[$B_z$]{
				\includegraphics[width=0.309\linewidth]{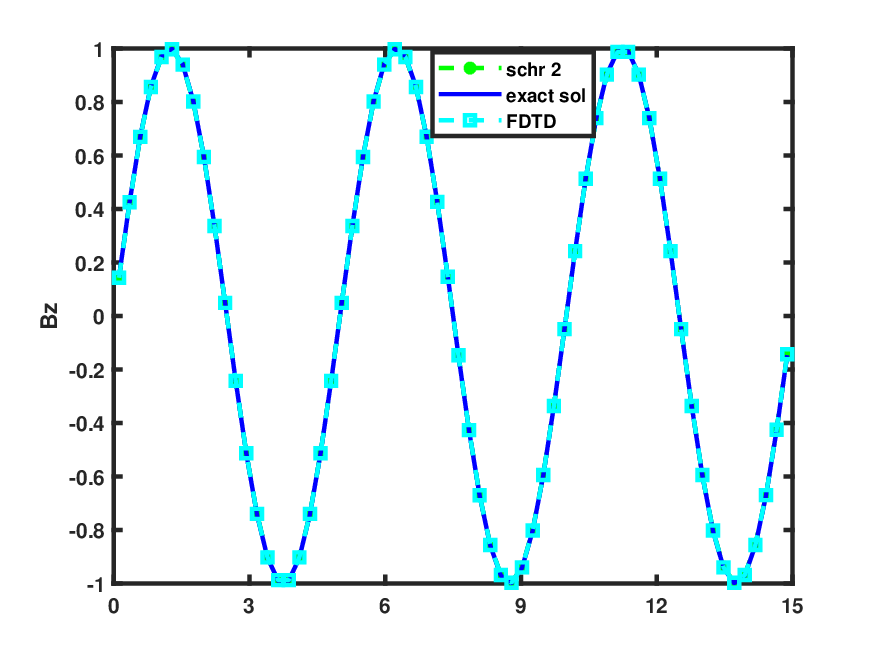}\label{fig:Bz_pec_FDTD}}
			\caption{Results for $T=1$ with  perfect conductor boundaries. The first row  are the simulations using upwind algorithm and the second row using Yee's algorithm.}\label{fig:PEC_FDTD}
		\end{figure}

	         In the 1D case, the impedance boundary condition reads
	         \begin{equation}
	         	E_y(0)+vB_z(0)=0,\quad
	         	vB_z(L)-E_y(L)=0.
	         \end{equation}
           We set the source term and the initial value carefully such that the solution satisfies 
            \begin{equation}
            	E_x = -\sin(\pi(x+t)/5),\quad
            	E_y = 5\cos(\pi x/5)/\pi,\quad
            	B_z = t\sin(\pi x/5) - 5/\pi.
            \end{equation}
           Similar to \eqref{eq:PE left}, one gets the left impedance boundary condition of the domain for Equation~\eqref{eq:tilde psi},
           \begin{equation}
           		\begin{bmatrix}
           			\tilde{\psi}_0\\
           			\tilde{\psi}_2\\
           			\tilde{\psi}_4\\
           			\tilde{\psi}_6
           		\end{bmatrix}
           	=\frac{1}{2\sqrt{2}}\begin{bmatrix}
           		1 &1  &-\frac{v-1}{v+1}  &\frac{v-1}{v+1}\\
           		-1 &-1  &-\frac{v-1}{v+1}  &\frac{v-1}{v+1}\\
           		-\frac{v-1}{v+1}  &\frac{v-1}{v+1}  &1 &1\\
           		-\frac{v-1}{v+1}  &\frac{v-1}{v+1}  &-1 &-1\\
           		\end{bmatrix}
           	\begin{bmatrix}
           		\tilde{\psi}_1\\
           		\tilde{\psi}_3\\
           		\tilde{\psi}_5\\
           		\tilde{\psi}_7
           	\end{bmatrix}
           =B_{E2O}           	\begin{bmatrix}
           	\tilde{\psi}_1\\
           	\tilde{\psi}_3\\
           	\tilde{\psi}_5\\
           	\tilde{\psi}_7
           \end{bmatrix}.
           \end{equation}
          By replacing the matrix $B_{E2O}$ in \eqref{eq:AD IMPE}, one obtains the ODE system for the quantum simulation. 
          	\begin{figure}[htbp]
          	\centering
          	\subfigure[$E_x$]{
          		\includegraphics[width=0.309\linewidth]{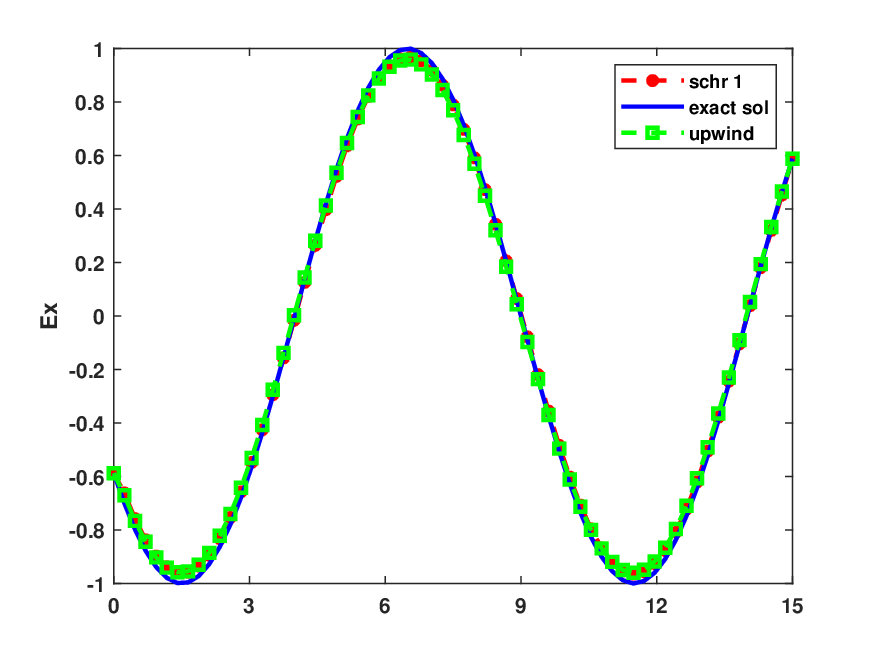}\label{fig:Ex_IM_upwind}}
          	\subfigure[$E_y$]{
          		\includegraphics[width=0.309\linewidth]{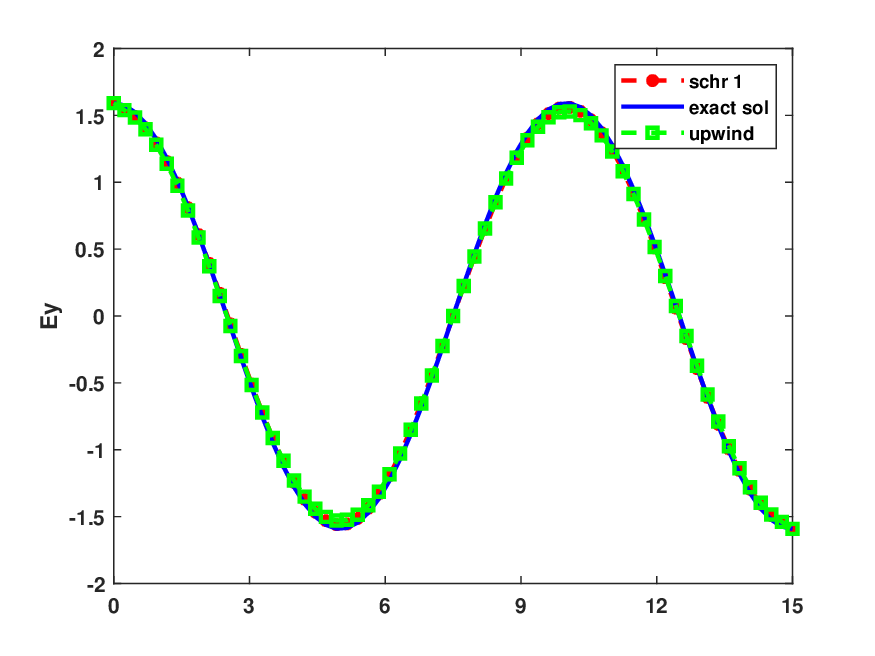}\label{fig:Ey_IM_upwind}}
          	\subfigure[$B_z$]{
          		\includegraphics[width=0.309\linewidth]{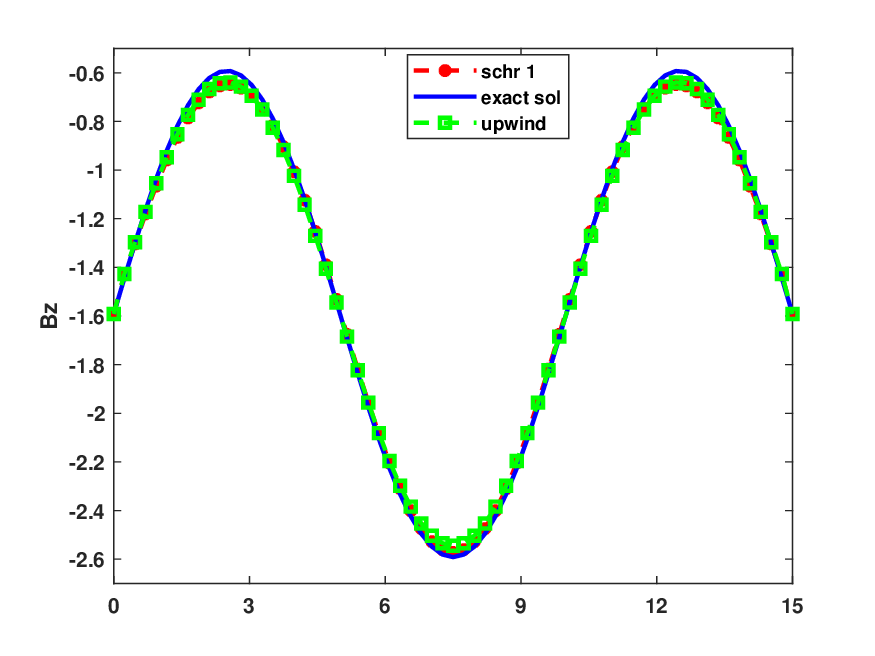}
            \label{fig:Bz_IM_upwind}}
          	\subfigure[$E_x$]{
          		\includegraphics[width=0.309\linewidth]{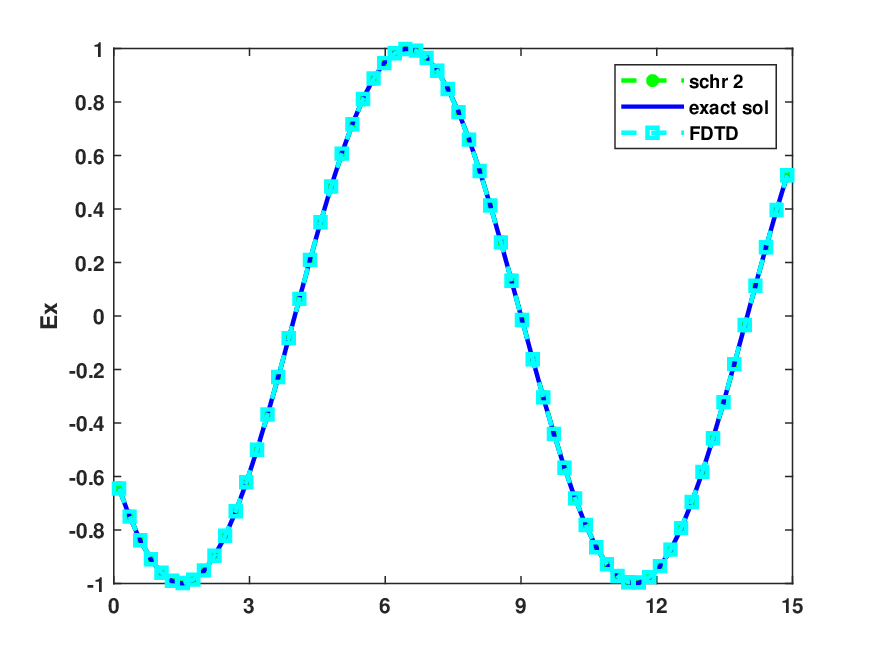}\label{fig:Ex_IM_FDTD}}
          	\subfigure[$E_y$]{
          		\includegraphics[width=0.309\linewidth]{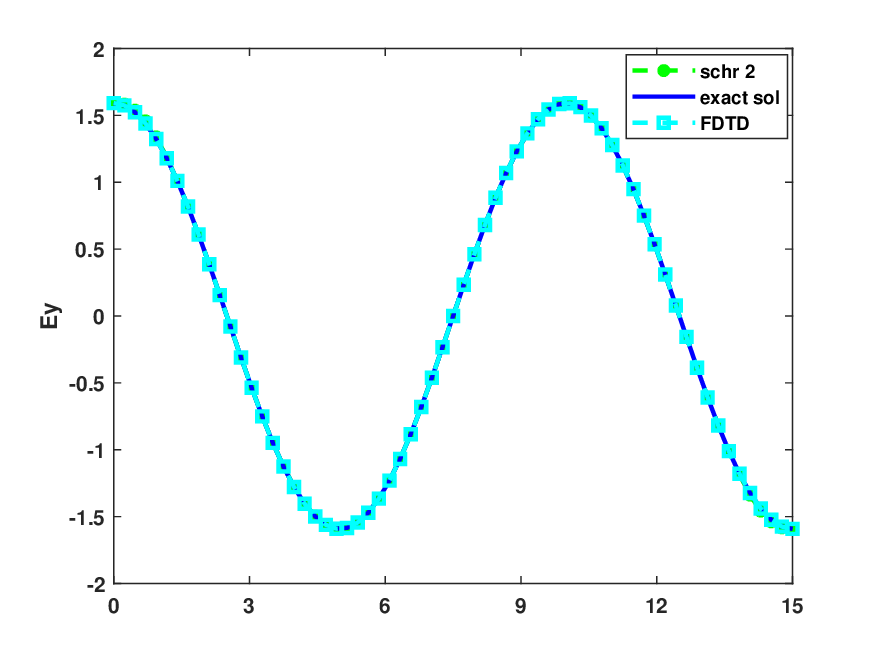}\label{fig:Ey_IM_FDTD}}
          	\subfigure[$B_z$]{
          		\includegraphics[width=0.309\linewidth]{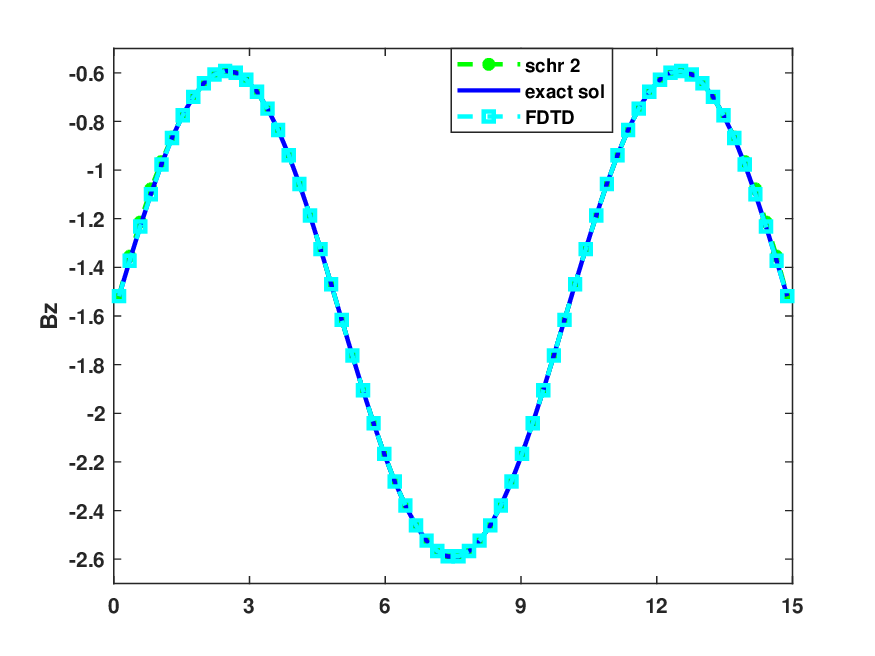}\label{fig:Bz_IM_FDTD}}
          	\caption{Results for $T=1$ with  impedance boundary conditions. The first row  are the simulations using upwind algorithm and the second row using Yee's algorithm.}\label{fig:IM_FDTD}
          \end{figure}
      
          Comparing with the figures in  Fig.~\ref{fig:PEC_FDTD} and Fig.~\ref{fig:IM_FDTD},  it can be seen that Schr$\ddot{\text{o}}$dingerisa-tion with both
          upwind algorithm and Yee's method can match the boundary conditions very well.
          Since the upwind scheme only has a first order accuracy of spatial discretization, the later is closer to the exact solution while using the same mesh size.
        \subsection{Simulations in inhomogeneous medium}
         Here, we simulate the evolution of a Gaussian pulse form a medium with $\varepsilon_1 =1$, $\mu_1=1$  to the other with $\varepsilon_2=3$, $\mu_2=1$, the discontinue parameter $\varepsilon$ and its approximation are shown in 
         Fig.~\ref{fig:dis_nx}. 
         The initial electromagnetic pulse is given by
         \begin{equation}
         	E_y(x,0)=B_z(x,0) = 0.01\exp\big[-\frac{20(x-3)^2}{100}\big].
         	\end{equation}
          In Fig.~\ref{fig:delta}, the first row are the results of Schr$\ddot{\text{o}}$dingerisation approach with $M=N=2^7$ and the second
          row are the approximation of QLA \cite{va20,va202} with fine mesh $M=2^{11}$,
          which shows that the numerical solutions from Schr$\ddot{\text{o}}$dingerisation
          method are in agreement with the ones of QLA .
          
    \begin{figure}[t]
	\centering
	\subfigure[$t=5$]{
		\includegraphics[width=0.309\linewidth]{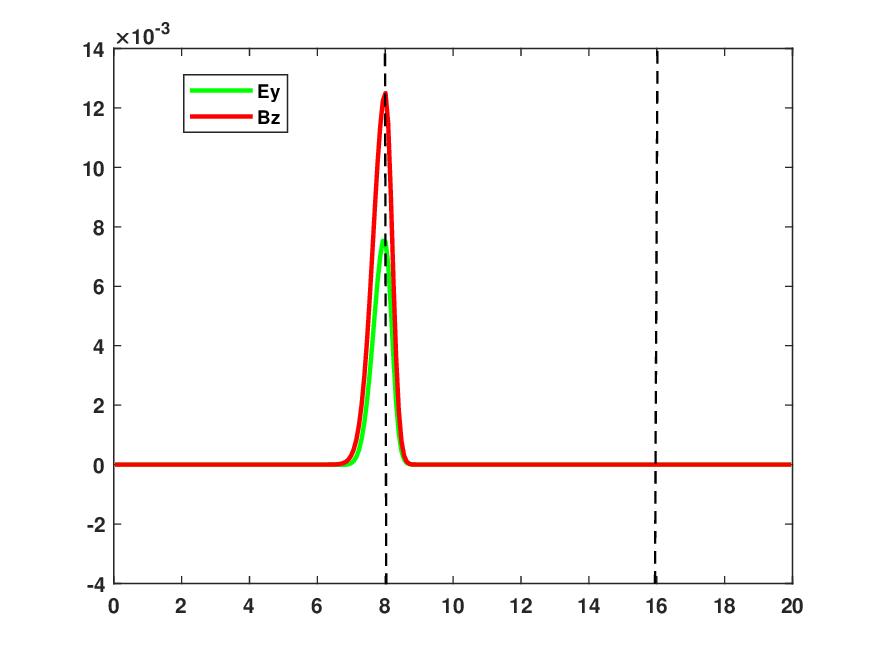}\label{fig:delta_Ep}}
	\subfigure[$t=10$]{
		\includegraphics[width=0.309\linewidth]{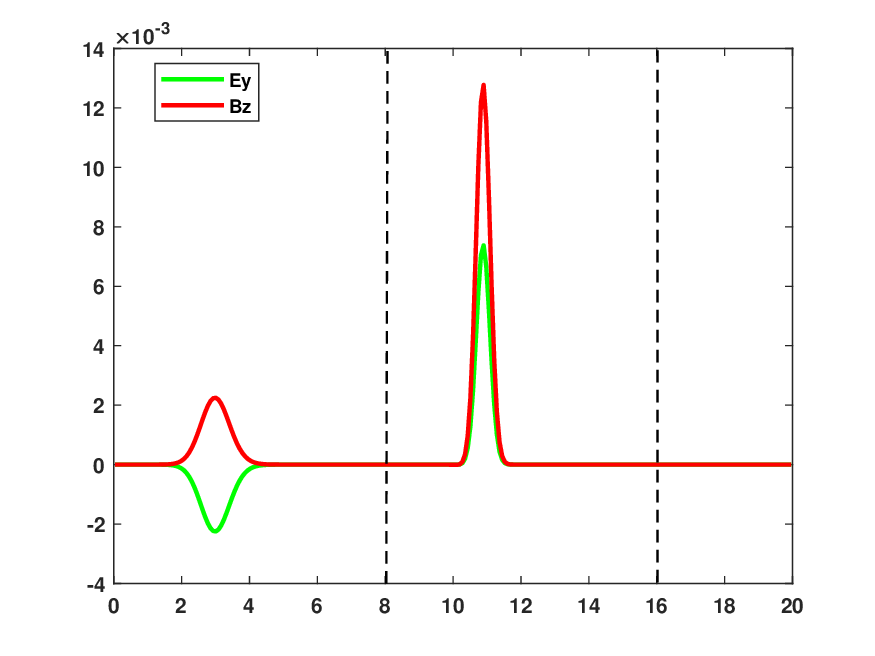}\label{fig:delta_Ep_T10}}
	\subfigure[$t=19.5$]{
		\includegraphics[width=0.309\linewidth]{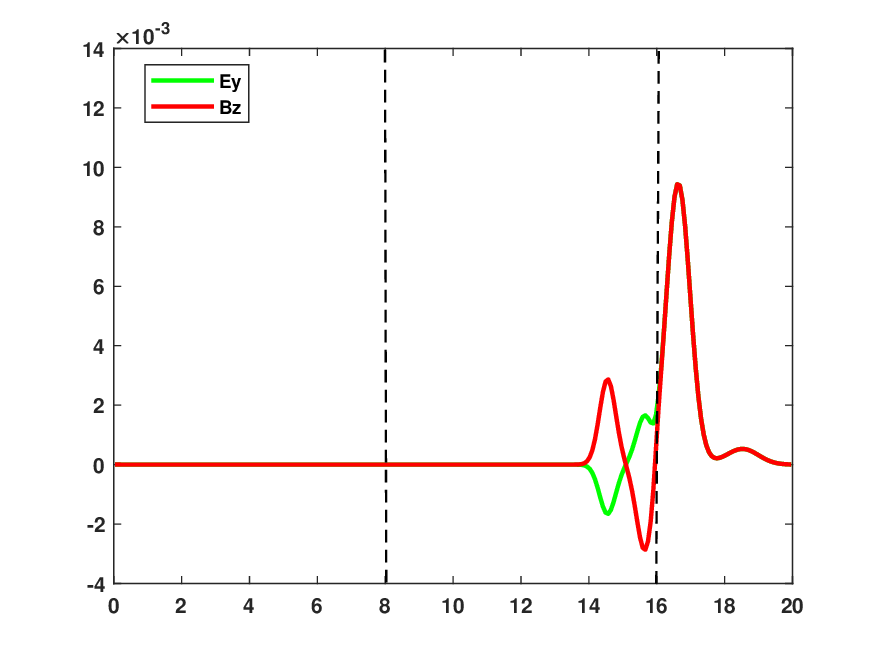}\label{fig:delta_Ep_T19.5}}
	\subfigure[$t=5$]{
		\includegraphics[width=0.309\linewidth]{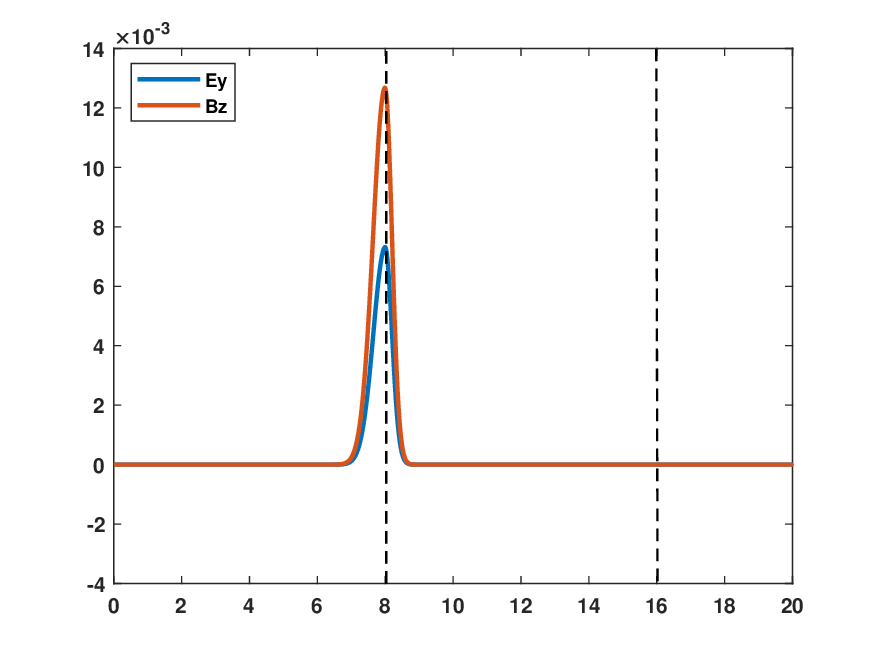}\label{fig:delta_QLA}}
	\subfigure[$t=10$]{
		\includegraphics[width=0.309\linewidth]{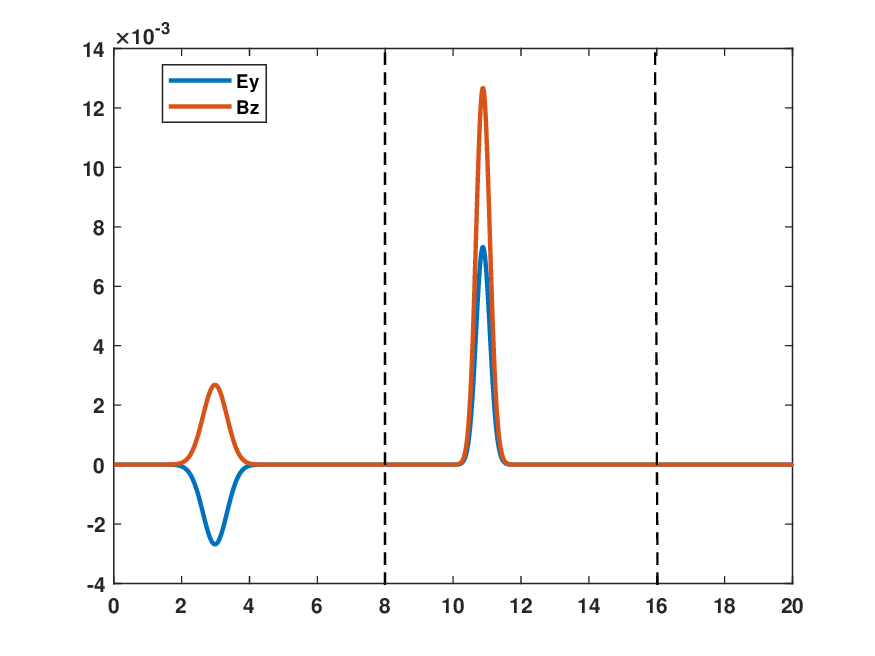}\label{fig:delta_QLA_T10}}
	\subfigure[$t=19.5$]{
		\includegraphics[width=0.309\linewidth]{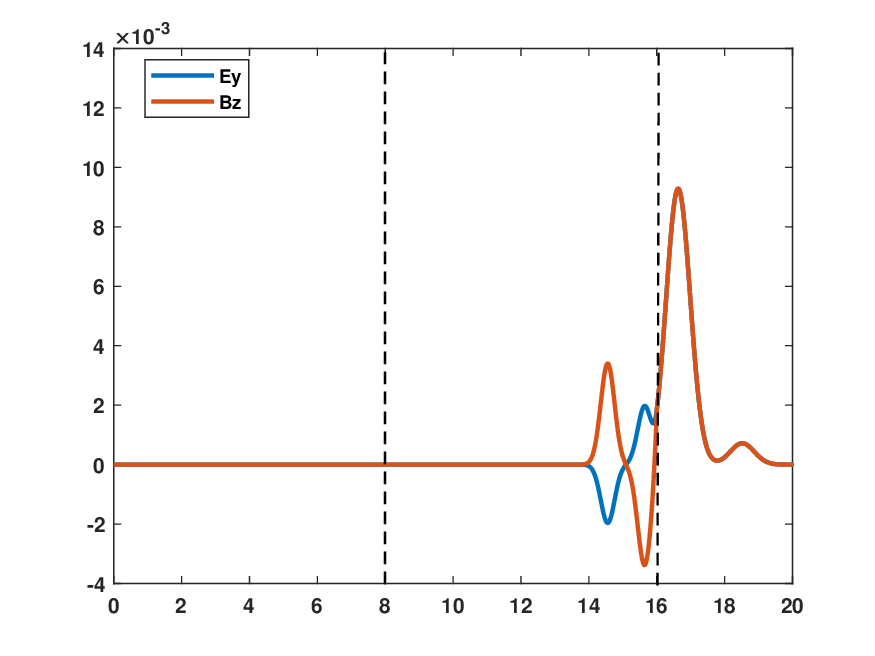}\label{fig:delta_QLA_T19}}
	\caption{Approximation of electromagnetic fields with schr 1 and QLA algorithm .}\label{fig:delta}
   \end{figure}

         Next we verify the accuracy of the Schr$\ddot{\text{o}}$dingerisation approach combined with the IIM method for Equation~\eqref{eq:tilde psi}  with the jump conditions in \eqref{eq:tilde jump condition} . The medium is divided into two parts--the medium on the left
         $(-9\leq x\leq 0)$ is vacuum ($\varepsilon_1=\mu_1=1$),  the right $(0\leq x\leq 10)$ is a dielectric with $\varepsilon_2=\mu_2=2$ \cite{Lor88}.
         The incident plane  is 
         \begin{equation}
         	E_{y,\text{inc}}=\exp{[i(\omega t-k_1 x)]}, \quad
            B_{z,\text{inc}} =\exp{[i(\omega t-k_1 x)]}/v_1,		
         \end{equation}
     where $\omega=0.5$, $v_1=1/\sqrt{\varepsilon_1 \mu_1}$, $k_i =\omega\sqrt{\varepsilon_i \mu_i}$ and $Z_i=\sqrt{\mu_i/\varepsilon_i}$, $i=1,2$. A reflective and transmitted wave 
     will be generated when the incident wave encounters the interface.
     The boundary conditions are chosen such that the exact solutions satisfy
     \begin{align*}
         E_y(x,t)&=\begin{cases}
             \exp{[i(\omega t-k_1x)]}+\frac{Z_2-Z_1}{Z_1+Z2}\exp{[i(\omega t+k_1x)]} \quad x<0\\
             \frac{2 Z_2}{Z_1+Z_2} \exp{[i(\omega t-k_2 x)]}  \quad x>0,
         \end{cases}\\
         B_z(x,t)&=\begin{cases}
             \frac{1}{v_1}\exp{[i(\omega t-k_1x)]}
             -\frac{Z_2-Z_1}{v_1(Z_1+Z2)}\exp{[i(\omega t+k_1x)]} \quad x<0,\\
             \frac{2\mu_2}{Z_1+Z_2} \exp{[i(\omega t-k_2 x)]} \quad
             x>0,
         \end{cases}
     \end{align*}
     The mesh size is set by $M=N=2^7$. As shown in Fig.~\ref{fig:dis_EB}, the Schr$\ddot{\text{o}}$dingerisation method captures the jump conditions across the material interface well, and its approximation is close to the exact solution.
\begin{figure}[t]
	\centering
	\subfigure[$E_y$]{
		\includegraphics[width=0.319\linewidth]{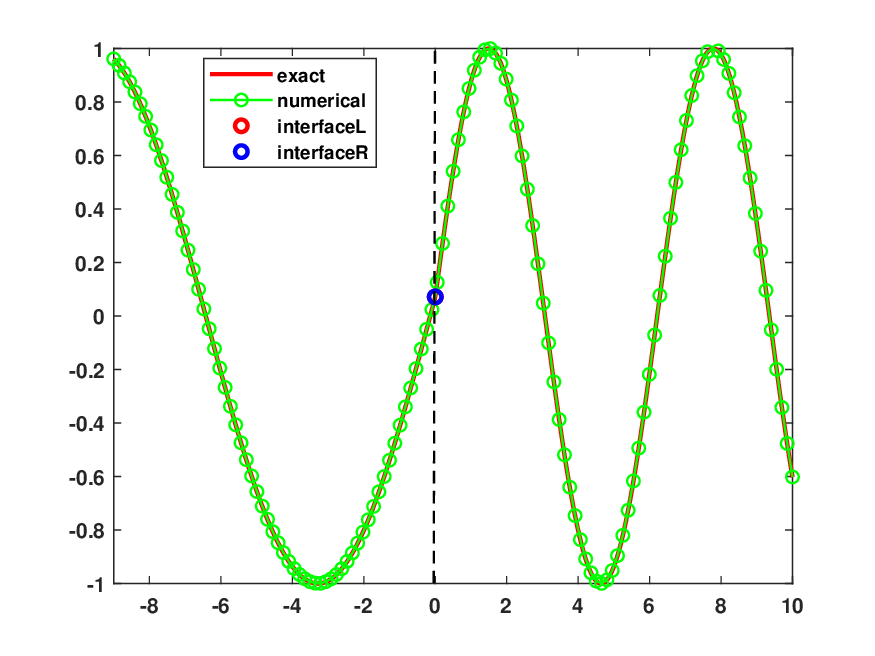}\label{fig:Ez_real}
		\includegraphics[width=0.319\linewidth]{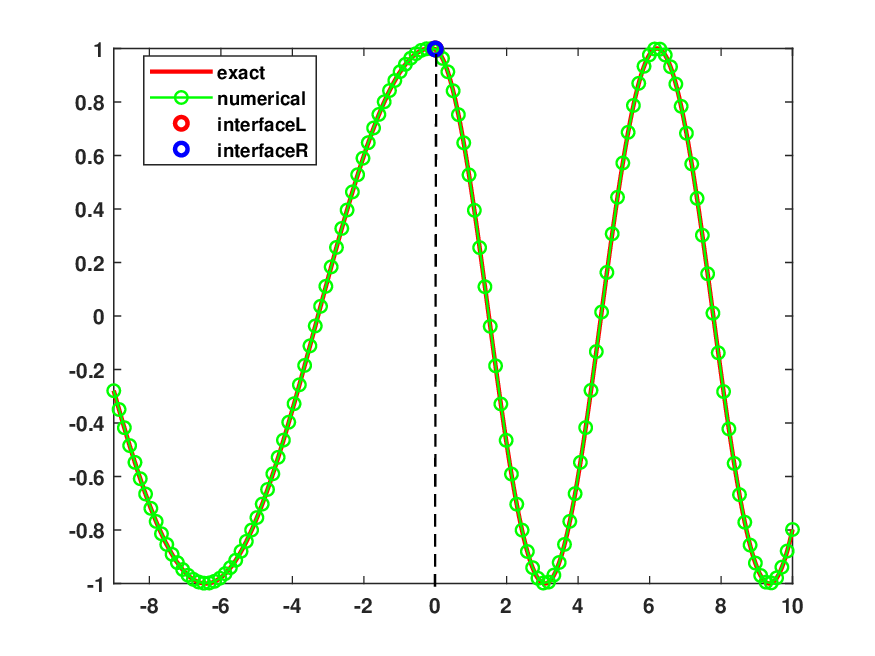}\label{fig:Ez_imag}}
	\subfigure[$B_z$]{
	    \includegraphics[width=0.319\linewidth]{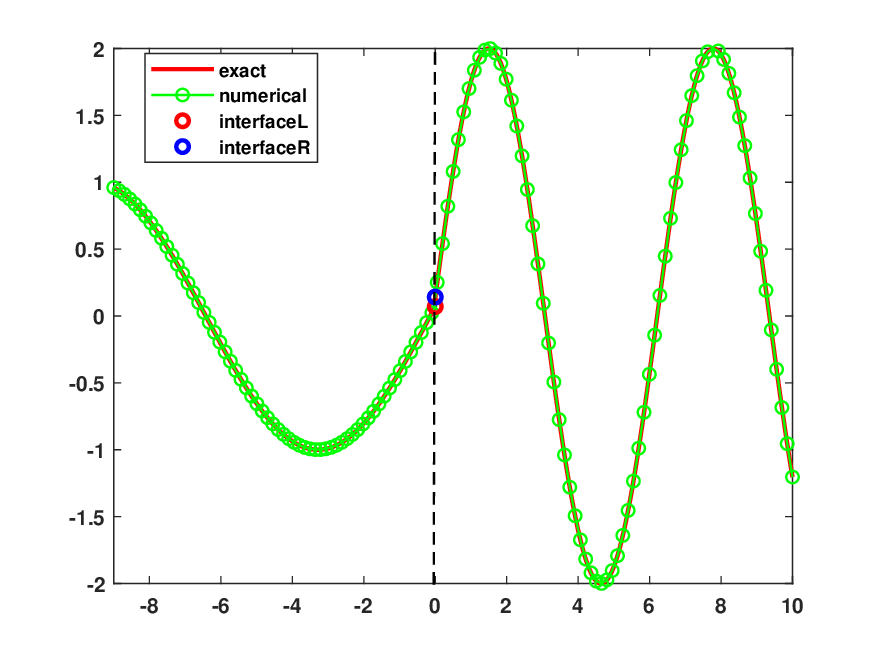}\label{fig:Bz_real}
        \includegraphics[width=0.319\linewidth]{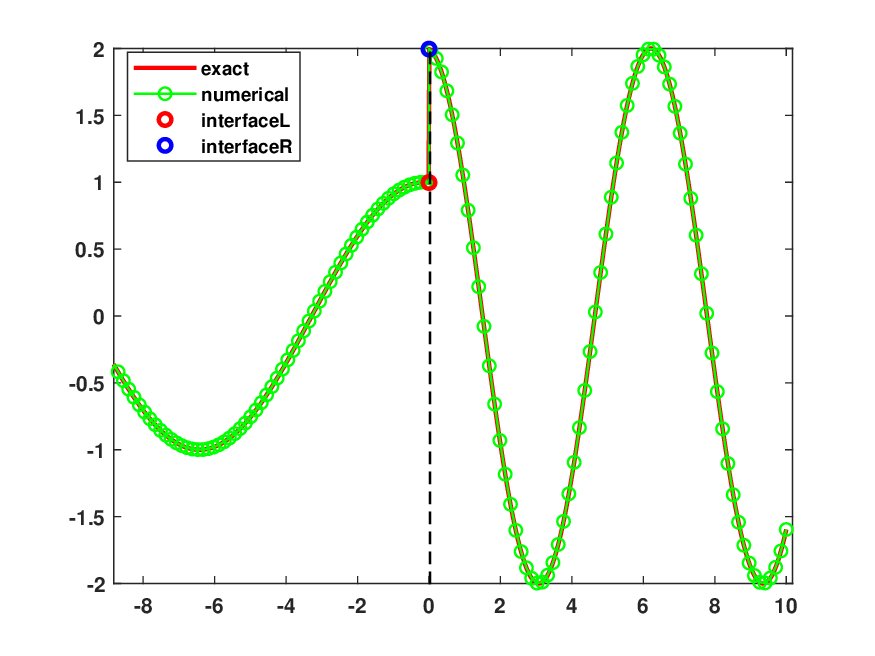}\label{fig:Bz_imag}}
	\caption{Electromagnetic fields at $T=3$.}
	\label{fig:dis_EB}
\end{figure} 

\section{Conclusions}
In this paper, we propose quantum algorithms for Maxwell's equations, using the Schr$\ddot{\text{o}}$dingerisation method introduced in \cite{JLY22a,JLY22b}. The proposed method has been demonstrated for the eight-dimensional matrix representation of Maxwell's equations based on the Riemann-Silberstein  vectors, and to electromagnetic models  based on  electric and magnetic fields. While source terms and physical boundary conditions are natural in simulations, quantum simulation incorporating these conditions are difficult due to the lack of unitarity of these systems.  We give implementation details for three physical boundary conditions, including periodic, perfect conductor and impedance boundary conditions. In addition, we simulate Maxwell's equations in a linear inhomogeneous medium with interface conditions.  Finally, we touch upon continuous-variable quantum systems to simulate Maxwell's equations via Schr$\ddot{\text{o}}$dingerisation.

 We did not consider the treatment of quantum simulations
algorithms for Maxwell's equations in unbounded domains. 
As pointed out in \cite{JLLY22}, the Schr$\ddot{\text{o}}$dingerisation  method can be applied to quantum dynamics with artificial boundary conditions. These will be the subject of our future research. 

\section*{Acknowledgement}

SJ was partially supported by the NSFC grant No. 120310-13, 
the Shanghai Municipal Science
and Technology Major Project (2021SHZDZX010-2), and the Innovation Program of Shanghai Municipal Education Commission (No. 2021-01-07-00-02-E00087). NL acknowledges funding from the Science and Technology Program of Shanghai, China (21JC1402900).
Both SJ and NL are also supported by the Fundamental Research Funds for the Central
Universities.
CM was partially supported
by China Postdoctoral Science Foundation (No. 2023M732248).


			\end{document}